\documentclass{amsart}
\usepackage{amsmath}
\usepackage{amssymb}
\usepackage{bbm}
\usepackage{mathdots}
\usepackage{dsfont}
\usepackage{todonotes,lscape}
\usepackage{graphicx}
\usepackage{hyperref}
\usepackage{tikz}
\usetikzlibrary{calc}
\usepackage{tabularx}
\usepackage{subcaption}
\usepackage{comment}

\def\RR{{\mathbb R}}

\def\ZZ{{\mathbb Z}}

\def\cal{\mathcal}

\def\a{\alpha}
\def\b{\beta}

\def\l{\lambda}
\def\L{\Lambda}

\let\temp\phi
\let\phi\varphi
\let\varphi\temp

\def\ones{\mathbbm{1}}

\def\rk{\operatorname {rk}}
\def\Diag{\operatorname {diag}}

\def\spanset{\operatorname {span}}

\def\rank{\operatorname {rank}}
\def\Sp{\operatorname {Sp}}

\theoremstyle{definition}
\newtheorem{definition}{Definition}[section]
\newtheorem{example}[definition]{Example}

\theoremstyle{plain}
\newtheorem{theorem}[definition]{Theorem}
\newtheorem{lemma}[definition]{Lemma}
\newtheorem{proposition}[definition]{Proposition}
\newtheorem{corollary}[definition]{Corollary}
\newtheorem*{corollary*}{Corollary}

\newtheorem*{principle}{Principle}

\begin{document}

\title{Spectrahedral Geometry of Graph Sparsifiers}

\author[]{Catherine Babecki}
\address{Department of Mathematics, University of Washington, Seattle, WA 98195, USA} 
\email{cbabecki@uw.edu }
\email{steinerb@uw.edu}
\email{rrthomas@uw.edu}

\author[]{Stefan Steinerberger}

\author[]{Rekha R. Thomas}

\begin{abstract}
    We propose an approach to graph sparsification based on the idea of preserving the smallest $k$ eigenvalues and eigenvectors of the Graph Laplacian. This is motivated by the fact  that small eigenvalues and their associated eigenvectors tend to be more informative of the global structure and geometry of the graph than larger eigenvalues and their eigenvectors. 
    The set of all weighted subgraphs of a graph $G$ that have the same first $k$ eigenvalues (and eigenvectors) as $G$ is the intersection of a polyhedron with a cone of positive semidefinite matrices. We discuss the geometry of these sets and deduce the natural scale of $k$. Various families of graphs illustrate our construction.
\end{abstract}


\maketitle

\tableofcontents

\section{Introduction and Definition}
\subsection{Motivation} 
The purpose of this paper is to introduce a new type of graph sparsification motivated by spectral graph theory.
Let $G=([n],E,w)$ be a connected, undirected, positively weighted graph with vertex set $[n]$, edge set $E$ and weights $w_{ij} > 0$ for all $(i,j) \in E$. The Laplacian of $G$ is the weakly diagonally dominant positive semidefinite 
matrix $L_G$ defined as 
\begin{align}
(L_G)_{ij} = \left\{ 
\begin{array}{ll}
-w_{ij} & \textup{ if } i \neq j \textup{ and } (i,j) \in E \\
0 & \textup{ if } i \neq j \textup{ and } (i,j) \not \in E \\
\sum_{\{l \,:\, (i,l) \in E \}} w_{il} & \textup{ if } i = j
\end{array}
\right.
\end{align}
with eigenvalues $\lambda_1 \leq \lambda_2 \leq \ldots \leq \lambda_n$. The all-ones vector $\ones$ shows that the smallest eigenvalue of $L_G$ is $\lambda_1 = 0$ with eigenvector $\phi_1 = \ones$. The lower end of the spectrum of $L_G$ is sometimes referred to as the {\em low frequency} eigenvalues of $G$. 

Spectral graph theory is broadly concerned with how the structure of  $G$ is related to the 
spectrum of $L_G$. 
 For example, $G$ has $k$ connected components if and only if $\lambda_1 = \dots = \lambda_k = 0 < \lambda_{k+1}$. In this paper we will only consider connected graphs; in this case, $\lambda_2 > 0$ plays an important role. This second eigenvalue is also known as the `algebraic connectivity' of $G$ and serves as a quantitative measure of how connected $G$ is. A fundamental result in spectral graph theory is {\em Cheeger's inequality} \cite{cheeger}, which connects 
 $\lambda_2$ to the density of the sparsest cut in $G$.
  Eigenvectors associated to small eigenvalues also carry important information.
 For example, these eigenvectors can be used to produce an approximate Euclidean embedding of a graph (see \S \ref{sec:conc} for the connection to dimensionality reduction).
 The common theme that underlies all these results can be summarized in what we will refer to as the Spectral Graph Theory Heuristic.

\begin{quote}
\textbf{Spectral Graph Theory Heuristic.} The low-frequency eigenvalues (and eigenvectors) of $L_G$ capture the global structure of $G$.
\end{quote}

Another way of interpreting this heuristic is by saying that eigenvectors corresponding to small eigenvalues tend to change very little across edges -- they are `smooth' over the graph and capture global structure. In contrast, eigenvectors corresponding to very large eigenvalues tend to oscillate rapidly and capture much more local phenomena (see, for example, \cite{steinosc}). All of these results point to the low frequency portion of the Laplacian spectrum as the `fingerprint' of $G$ that captures the global structure of $G$ (while the high frequency portion adds finer detail).

Sparsifying a graph $G$ is the 
process of modifying the weights on edges (or removing them altogether) while preserving 
`essential' properties of $G$. This raises the question  of which properties are essential, and how to sparsify in a way that preserves these properties. 
Motivated by the Spectral Graph Theory Heuristic, we 
introduce a Spectral Sparsification Heuristic (see \S \ref{sec:conc} for additional motivation). 

\begin{quote}
\textbf{Spectral Sparsification Heuristic.} 
A sparsification of $G$ should preserve the first $k$ eigenvalues and eigenvectors of $L_G$.
\end{quote}

The motivation is clear from what has been said above about the interplay between the eigenvalues of $L_G$ and the properties of $G$.
As one would expect, the parameter $k \in \mathbb{N}$ in the heuristic will determine a natural scale. Small values of $k$ only ask for the preservation of relatively few eigenvectors and eigenvalues which allows for a larger degree of sparsification. If one requires that many eigenvectors and eigenvalues remain preserved, then this will restrict how much sparsification one could hope for. There is a natural linear algebra heuristic (see \S \ref{sec:lin}) suggesting a value of $k_0$ such that when $k > k_0$, the only $k-$sparsifier of $G$ is $G$ itself. We now make these notions precise. 

\subsection{Bandlimited Spectral Sparsification}
Let $G$ be a connected graph and let the eigenvalues of $L_G = D-A$ be 
$0 = \lambda_1 < \lambda_2 \leq \cdots \leq \lambda_n$. All graphs in this paper are connected and 
hence the eigenspace of $\lambda_1=0$ is spanned by $\ones$, the all-ones vector. 
Fix an orthonormal eigenbasis $\{ \phi_1 = \ones/\sqrt{n}$, $\phi_2, \dots, \phi_n \}$ of $L_G$   
so that $L_G \phi_i = \lambda_i \phi_i$ for all $i=1,\ldots,n$. Then the spectral decomposition of $L_G$ is 
\begin{align}
L_G = \Phi \Lambda \Phi^\top = \sum_{i=2}^n \lambda_i \phi_i \phi_i^\top.
\end{align} 
where 
\begin{align}
    \Lambda := \Diag(0,\lambda_2, \ldots, \lambda_n) \quad \textup{ and } \quad
    \Phi := \begin{bmatrix} \phi_1 & \phi_2 & \cdots & \phi_n \end{bmatrix} \in \mathbb{R}^{n \times n}.
\end{align}

We denote the eigenpairs of a graph $G$ as $(\l_i, \phi_i)$. 
Fixing the above notation, we can now give a formal definition of sparsification.

\begin{definition} \label{def:k-sparsifier}
       \begin{enumerate}
    \item A subgraph $\tilde{G}=([n],\tilde{E}, \tilde{w})$ of $G$ where $\tilde{E} \subseteq E$ and $\tilde{w} > 0$ is $k-${\bf isospectral} to $G$ if $\tilde{G}$ and $G$ have the same first $k$ eigenvalues and eigenvectors, i.e., the first $k$ eigenpairs of $\tilde{G}$ are $(0, \phi_1), \ldots, (\lambda_k,\phi_k)$. 
    \item A $k-$isospectral subgraph $\tilde{G}=([n],\tilde{E}, \tilde{w})$ of $G$ is a $k-${\bf sparsifier} if $\tilde{E} \subsetneq E$, i.e., at least one edge of $G$ does not appear in $\tilde{G}$.
    \end{enumerate}
\end{definition}

While the eigenvalues are unique, there is a choice of eigenbasis $\{\phi_1, \ldots, \phi_n\}$. The above definition is made with respect to a fixed eigenbasis of $L_G$. We will see that there is no ambiguity if we choose to preserve all eigenpairs in a given eigenspace, in particular, if the eigenvalues of $L_G$ do not have multiplicity. 
We illustrate Definition~\ref{def:k-sparsifier}
on a simple example in Figure~\ref{fig: lollipop ex}. 

\begin{center}
    \begin{figure}[h!]
\includegraphics[scale = .4]{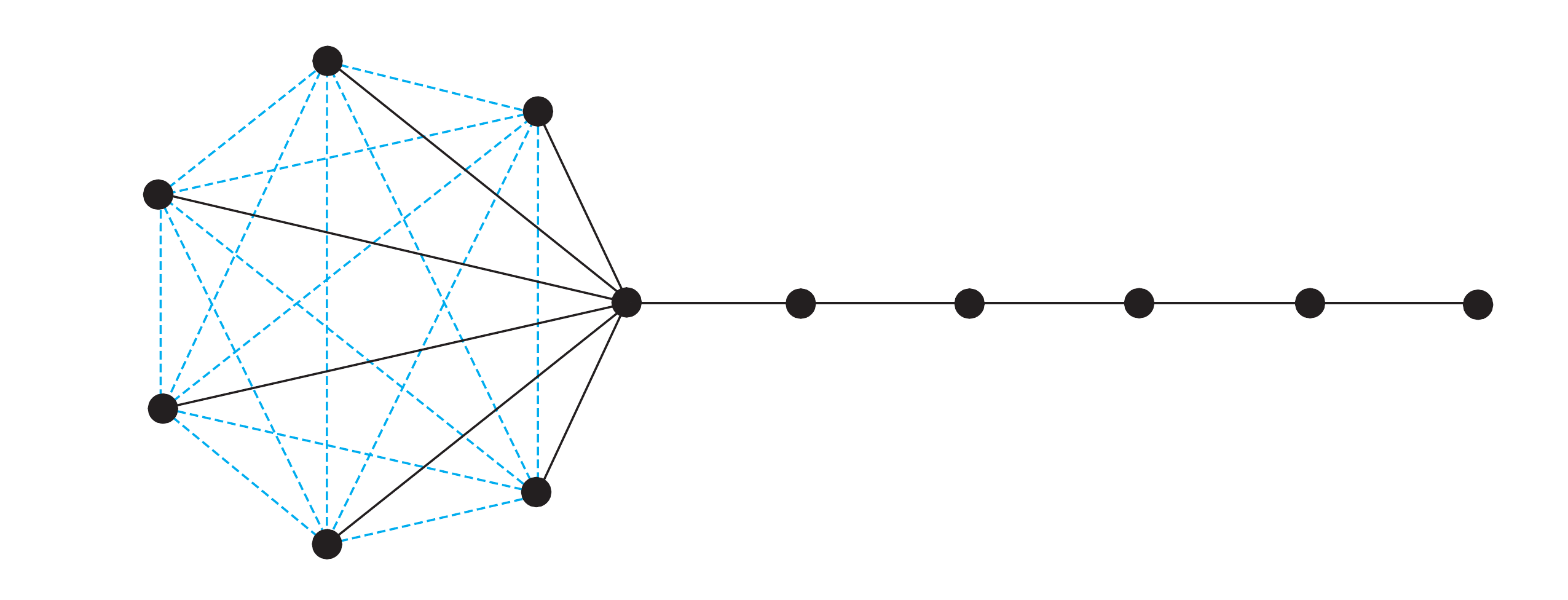}
    \caption{The lollipop graph on $n=12$ vertices and $m=26$ edges can be sparsified to a spanning tree while keeping the first $k=3$ eigenpairs unchanged. } \label{fig: lollipop ex}
        \end{figure}
\end{center}

We note that graphs which are $1-$isospectral to $G$ are not interesting -- these include every subgraph of $G$, even the empty graph and need not preserve any of the desirable structure of $G$. 
For this reason, we restrict our attention to $k-$isospectral subgraphs and $k-$sparsifiers for $k \geq 2$.

\begin{lemma} For $k \geq 2$, a $k-$isospectral subgraph $\tilde{G}=([n],\tilde{E}, \tilde{w})$ of a connected graph $G=([n],E,w)$ is connected (and hence spanning). \label{lem: cant disconnect}
\end{lemma}

This follows at once from the fact that $\lambda_2 > 0$ if and only if the graph is connected. For $k \geq 2$, any $k-$isospectral subgraph $\tilde{G}$ preserves $\l_2(\tilde{G}) = \l_2(G) > 0$.

Recall that the quadratic form associated to $L_G$ is 
$Q_G(x) := x^\top L_G x$. 
An immediate consequence of Definition~\ref{def:k-sparsifier} is that if $\tilde{G}$ is a 
$k-$isospectral subgraph of $G$, then its quadratic form $Q_{\tilde{G}}(x)$ agrees with $Q_G(x)$ on the span of 
the $k$ eigenvectors of $G$ that we are fixing. 
\begin{lemma} \label{lem:low spec equality implies quadratic form equality}
    If $\tilde{G}$ is a $k-$isospectral subgraph of $G$ then 
    $$Q_G(x) = Q_{\tilde{G}}(x) \qquad \mbox{for all}~ x \in \textup{span}\{ \phi_1, \ldots, \phi_k\}.$$
\end{lemma}
The \emph{ \nameref{pf:low spec equality implies quadratic form equality}} (and all other proofs) are in \S~\ref{sec:proofs}. 
We note that the converse of Lemma~\ref{lem:low spec equality implies quadratic form equality} is false, namely that if 
$G'$ is a subgraph of $G$ with the property that $Q_{G'}(x) = Q_G(x)$ for all $x \in \textup{span}\{ \phi_1, \ldots, \phi_k\}$ then it does not 
mean that the smallest $k$ eigenpairs of $G'$ and $G$ agree (see Example~\ref{ex: cube Qk(x)}).  Therefore, the notion of 
$k-$isospectrality is stronger than asking for the quadratic forms to agree on the span of the low frequency eigenvectors of $G$.

\subsection{Summary of Results and Organization of Paper}
We start by illustrating the concepts of $k-$isospectral graphs and $k-$sparsifiers in Section \ref{sec:2 examples} on two small and completely explicit examples, which amounts to some concrete computations. 

Section \ref{sec:structure theorem} contains our main result, the Structure Theorem \ref{subsec:structure theorem}. It describes the overall structure of $k-$isospectral graphs as a set of positive semidefinite (psd) matrices satisfying linear constraints coming from the requirements that we (a) should not create new edges and (b) want edge weights to be non-negative.  The search for $k-$sparsifiers with few edges corresponds to finding points in this set satisfying many inequality constraints at equality.

The heart of Section \ref{sec:lin} is a \textit{Linear Algebra Heuristic} saying that, \textit{generically}, if
$$ |E| \leq \binom{n}{2} - \binom{n-k+1}{2}$$ 
then the only $k-$sparsifier of $G$ is $G$ itself. This is not always true but is generically true (for example for random graphs) for reasons that will be explained. This suggests that dense graphs with many edges will usually lead to $k-$sparsifiers for reasonable large $k$ while graphs with few edges are not so easily sparsified (which naturally corresponds to our intuition). Section \ref{sec:lin} discusses a number of results in this spirit.  
The Structure Theorem explains the effectiveness of the Linear Algebra Heuristic: \textit{generically}, linear systems of equations are solvable when the number of variables is at least as large as the number of equations. 
Having analyzed the hypercube graph $Q_d$ on $\{0,1\}^d$  in Section \ref{sec:lin}, we illustrate our model of graph sparsification on several more graph families in Section \ref{sec:families}. 
Section \ref{sec:preserv} discusses an alternative notion of sparsification centered only around preserving the quadratic form $Q(x) = x^\top L_G x$ on low-frequency eigenvalues and eigenvectors. We contrast this model to our main model. 
Section \ref{sec:proofs} contains all the proofs. We conclude with a number of remarks in Section \ref{sec:conc}.
\begin{center}
    \begin{figure}[h!]
\includegraphics[scale = .7]{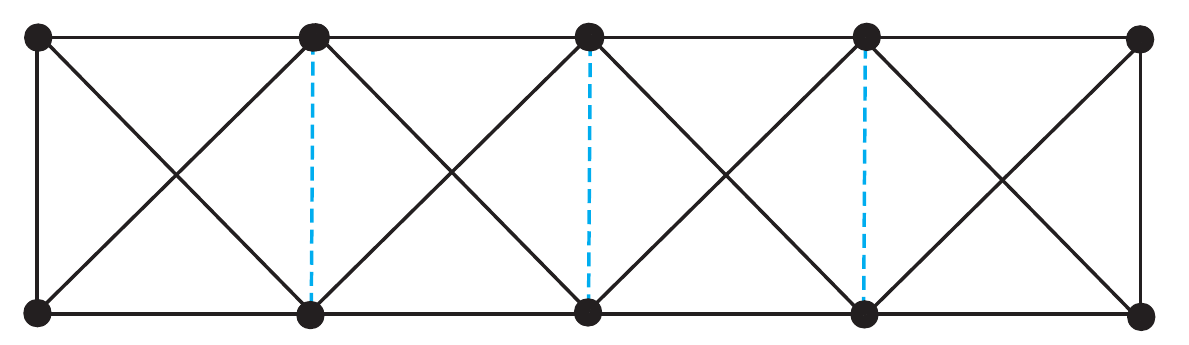}
    \caption{A  graph $G$ on $n=10$ vertices and $m=21$ edges. There exists a sparsification which removes the three dotted edges while preserving the first $k=5$ eigenvalues and eigenvectors of $G$.} 
        \end{figure}
\end{center}
\vspace{-.2 in}

\subsection{Related Work} 

The problem of graph sparsification is classical. The approach that is perhaps philosophically the closest to ours is that of spectral sparsification (Spielman \& Teng \cite{spiel1}). The main idea there is, given a graph $G$ and a tolerance $\varepsilon$, to find a graph $G'$ whose (weighted) edges are a subset of the edges of $G$ while, simultaneously, satisfying
$$ \forall~x \in \mathbb{R}^n \qquad  (1-\varepsilon) \cdot Q_{G'}(x)\leq 
 Q_G(x) \leq  (1+\varepsilon) \cdot Q_{G'}(x).
$$
The goal is thus to find a subset of the edges of $G$ such that (after changing their weights) the new graph $G'$ has a  Laplacian whose quadratic form behaves very similarly to that of $L_G$. We have
$$ 
Q_G(x) = \sum_{(i,j) \in E} w_{ij} (x(i) - x(j))^2.$$
Thus, one way of interpreting spectral sparsification is that one finds a sparser graph that assigns to each function $x:V \rightarrow \mathbb{R}$ a nearly identical level of smoothness. This notion has been very influential, see e.g. \cite{spiel4, spiel3, bhat, dec, spiel2}. 

Our approach is similar in spirit in the sense that it relies on spectral geometry as a measure of graph similarity, but it also differs in essential ways. 
\begin{enumerate}
    \item We do not even approximately preserve large eigenvalues of $L_G$. This is bolstered by the Spectral Graph Theory Heuristic that the low-frequency eigenpairs of $L_G$ are the heart of the graph. 
    \item However, we preserve the first $k$ eigenvalues and their associated eigenvectors {\em exactly}, again motivated by the Spectral Graph Theory Heuristic, producing subgraphs that exactly preserve many of the essential structural properties of $G$. 
    See \S~\ref{sec:conc} for more on this. 
    \item A major result of spectral sparsification is that there necessarily exists sparsifiers of $G$ with a small number of edges,  
    whereas in our notion, some graphs simply cannot be sparsified. A simple example is the Hamming cube graph $Q_d$ on $\left\{0,1\right\}^d$ which cannot be sparsified (see Theorem \ref{thm:cube}) once we require that the first nontrivial eigenspace is preserved. This is well aligned with the idea that certain graphs, especially those with extraordinary degrees of symmetry, are already as sparse as possible.
\end{enumerate}

Perhaps the common denominator of the two approaches is for the quadratic form to be preserved (exactly or approximately) on the low-frequency eigenvectors. By Lemma~\ref{lem:low spec equality implies quadratic form equality}, our notion of isospectrality subsumes this requirement even under exact preservation. We refer to 
\S~\ref{sec:conc} for further comments.


\section{Two Concrete Examples}
\label{sec:2 examples}

In this section, we explicitly construct the set of $k-$isospectral subgraphs of two small graphs, and 
 use these examples to highlight various aspects of the underlying geometry. The needed computations were done `by hand' and their details will become clear after the  main structure theorem is introduced in Section~\ref{sec:structure theorem}. 
We will record only the upper triangular part of a symmetric matrix and will denote the identical lower triangular entries with a --. 

\subsection{The Butterfly Graph.}
\label{ex:butterfly}
    Let $G$ be the {\em Butterfly Graph} in Figure~\ref{fig:original butterfly graph}. 

    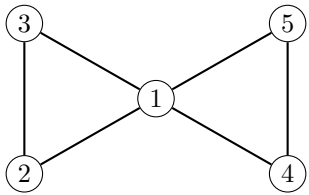
\begin{figure}[htbp]
        \centering
  \begin{tikzpicture}
        \tikzstyle{bk}=[circle, fill = white,inner sep= 2 pt,draw]
          \tikzstyle{red}=[circle, fill = red,inner sep= 2 pt,draw = red]
\node (v1) at (0,0) [bk] {$1$};
\node (v2) at (-1.75,-1) [bk] {$2$};
\node (v3) at (-1.75,1) [bk] {$3$};
\node (v4) at (1.75,-1)  [bk] {$4$};
\node (v5) at (1.75,1)  [bk] {$5$};
\node (e23) at (-2.2,0)  [] {};
\node (e45) at (2.2,0)  [] {};
\node (e12) at (-.8,-.8)  [] {};
\node (e13) at (-.8,.8)  [] {};
\node (e14) at (.8,-.8)  [] {};
\node (e15) at (.8,.8)  [] {};
\draw[thick] (v2) -- (v1) -- (v5);
\draw[thick] (v3) -- (v1) -- (v4);
\draw[thick]  (v3) -- (v2);
\draw[thick]  (v4) -- (v5);
\end{tikzpicture}
        \caption{The Butterfly Graph with weight $1$ on all edges.}
        \label{fig:original butterfly graph}
    \end{figure}
The eigenvalues of $L_G$ are $0,1,3,3,5$, and suppose we ask to preserve the first $k=2$ eigenvectors and eigenvalues. Since the first two eigenvalues, $0$ and $1$, both have multiplicity 1, there is a unique choice of eigenvectors (up to sign):
\begin{align*}
    \phi_1 = \frac{1}{\sqrt{5}} (1,1,1,1,1) \qquad \mbox{and} \qquad  \phi_2 = \frac{1}{2} (0,-1,-1,1,1). 
\end{align*}
Our goal is to find edge weights leading to (sub)graph Laplacians whose first eigenpair is $(0, \phi_1)$ (this is easily achieved: any connected subgraph of $G$ will do) and whose second eigenpair is $(1, \phi_2)$ (this is not quite so automatic).

 The (weighted) $2-$isospectral subgraphs of $G$ are indexed by the non-negative orthant of 
    the $ab$-plane as seen in Figure~\ref{fig:butterfly graph}, with each point $(a,b) \geq 0$  indexing the  subgraph $\tilde{G}$ of $G$ with edge weights
    $$ \tilde{w}_{12} = \tilde{w}_{13} = \tilde{w}_{14} = \tilde{w}_{15} = 1, \,\, \tilde{w}_{23} = a, \,\, \tilde{w}_{45} = b.$$
\begin{figure}[h]
        \centering
           \begin{tikzpicture}
        \tikzstyle{bk}=[circle, fill = white,inner sep= 2 pt,draw]
          \tikzstyle{red}=[circle, fill = red,inner sep= 2 pt,draw = red]
\node (v1) at (0,0) [bk] {$1$};
\node (v2) at (-1.75,-1) [bk] {$2$};
\node (v3) at (-1.75,1) [bk] {$3$};
\node (v4) at (1.75,-1)  [bk] {$4$};
\node (v5) at (1.75,1)  [bk] {$5$};
\node (e23) at (-2,0)  [] {$a$};
\node (e45) at (2,0)  [] {$b$};
\node (e12) at (-.8,-.8)  [] {$1$};
\node (e13) at (-.8,.8)  [] {$1$};
\node (e14) at (.8,-.8)  [] {$1$};
\node (e15) at (.8,.8)  [] {$1$};
\draw[thick] (v2) -- (v1) -- (v5);
\draw[thick] (v3) -- (v1) -- (v4);
\draw[dashed,thick,cyan]  (v3) -- (v2);
\draw[dashed,thick,cyan]  (v4) -- (v5);
\end{tikzpicture}
        \caption{The $2-$isospectral subgraphs of the Butterfly Graph are indexed by $(a,b) \geq 0$. The axes and origin of the non-negative $ab$-orthant index  the $2-$sparsifiers of $G$.}
               \label{fig:butterfly graph}
    \end{figure}
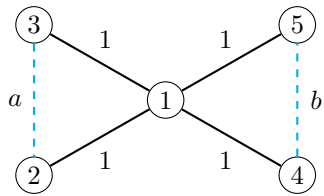

    On the $a$-axis we get the $2-$sparsifiers of $G$ that are missing the edge $(2,3)$, and on the $b$-axis we get the $2-$sparsifiers that are missing the edge $(4,5)$. 
 At $(0,0)$ we get the spanning tree $2-$sparsifier with edges $(1,2),(1,3),(1,4),(1,5)$ and all edge weights equal to $1$.
The set of Laplacians of all $2-$isospectral subgraphs of $G$ is
\[\textup{Sp}_G(2) = \left\{ L = 
\begin{bmatrix} 4 & -1 & -1 & -1 & -1 \\ 
    - & 1+a & -a & 0 & 0 \\ 
    - & - & 1+a & 0 & 0 \\
    - & - & - & 1+b & -b \\
    - & - & - & - & 1 + b
    \end{bmatrix}
    \,:\, a, b\geq 0\right\}.\]

 For example, if we choose $a=5/2, b= 5$, then we get a Laplacian with eigenvalues 
 $0,1,5,6,11$ and the same first two eigenvectors $\phi_1, \phi_2$ as $G$.

\subsection{The complete graph $K_4$} 
\label{ex:K4} 
For our second example, we consider the complete graph $K_4$ whose eigenvalues are $0,4,4,4$. We want to preserve the first $k=2$ eigenvectors and eigenvalues. Since the second eigenspace has a high multiplicity, we must specify which one-dimensional subspace of the three-dimensional eigenspace associated to eigenvalue $\lambda=4$ we wish to preserve. We choose
\begin{align*}
    \phi_1 = \frac{1}{2} (1,1,1,1) \qquad \mbox{and} \qquad  \phi_2 = \frac{1}{\sqrt{2}} (1,-1,0,0). 
\end{align*}

A computation shows that there indeed exists a nonempty set of $2-$isospectral graphs indexed by three parameters $a,b,c$ subject to seven inequalities. It is a fairly large space and contains all sorts of different $2-$isospectral graphs and $2-$sparsifiers, see Fig.~\ref{fig:K4 2-sparsifiers} for some examples. Fixing $\phi_2$  breaks the symmetry  of $K_4$ since we are deciding that this eigenvector is the most important one among all eigenvectors with eigenvalue $4$, and needs to be preserved. The weights can be read off the associated graph Laplacian.  The set of Laplacians of $2-$isospectral subgraphs of $K_4$ subject to fixing $\phi_1$ and $\phi_2$ is given explicitly by
$$
\Sp_{K_4}(2) =
\left\{ 
\begin{array}{r}
L = \begin{bmatrix}  
3 + \frac{a}{4} & \frac{a}{4}-1 & \frac{-a}{4} + \frac{c}{2 \sqrt{2}}-1 & \frac{-a}{4} + \frac{-c}{2 \sqrt{2}} -1\\
&&&\\
- & 3+\frac{a}{4} & \frac{-a}{4} + \frac{c}{2 \sqrt{2}} -1 & \frac{-a}{4} + \frac{-c}{2 \sqrt{2}}-1 \\
&&&\\
- & - & 3+\frac{a}{4}+\frac{b}{2}-\frac{c}{\sqrt{2}} & \frac{a}{4} - \frac{b}{2}-1\\
&&&\\
- & - & - & 3+\frac{a}{4}+\frac{b}{2}+\frac{c}{\sqrt{2}}
      \end{bmatrix} \,:\, \\ 
      \\
      a \leq 4, \,\, -a + \sqrt{2} c \leq 4, \,\,
      -a - \sqrt{2}c \leq 4,\,\,  a - 2b \leq 4\\
      \\
      a \geq 0,\,\, b \geq 0,\,\, ab \geq c^2
\end{array}
\right\}.
$$

Each $2-$isospectral subgraph of $K_4$ that preserves eigenvectors $\phi_1$ and $\phi_2$ is indexed by a 
point $(a,b,c)$ in the `boat-like' set shown in Figure~\ref{fig:K4new}. We note that this region is quite nontrivial and comprised of a mixture of curved surfaces (which come from the spectral requirement that a Laplacian is positive semidefinite) and flat hyperplanes (which come from linear inequalities defining non-negative edge weights). The set $\Sp_{K_4}(2)$ is unbounded in direction $(0,1,0)$ and hence the apparent ``back'' side of the boat-like shape is not actually there and the set extends infinitely far in the $b$ direction.
    \begin{figure}[h!]
        \includegraphics[scale=0.3]{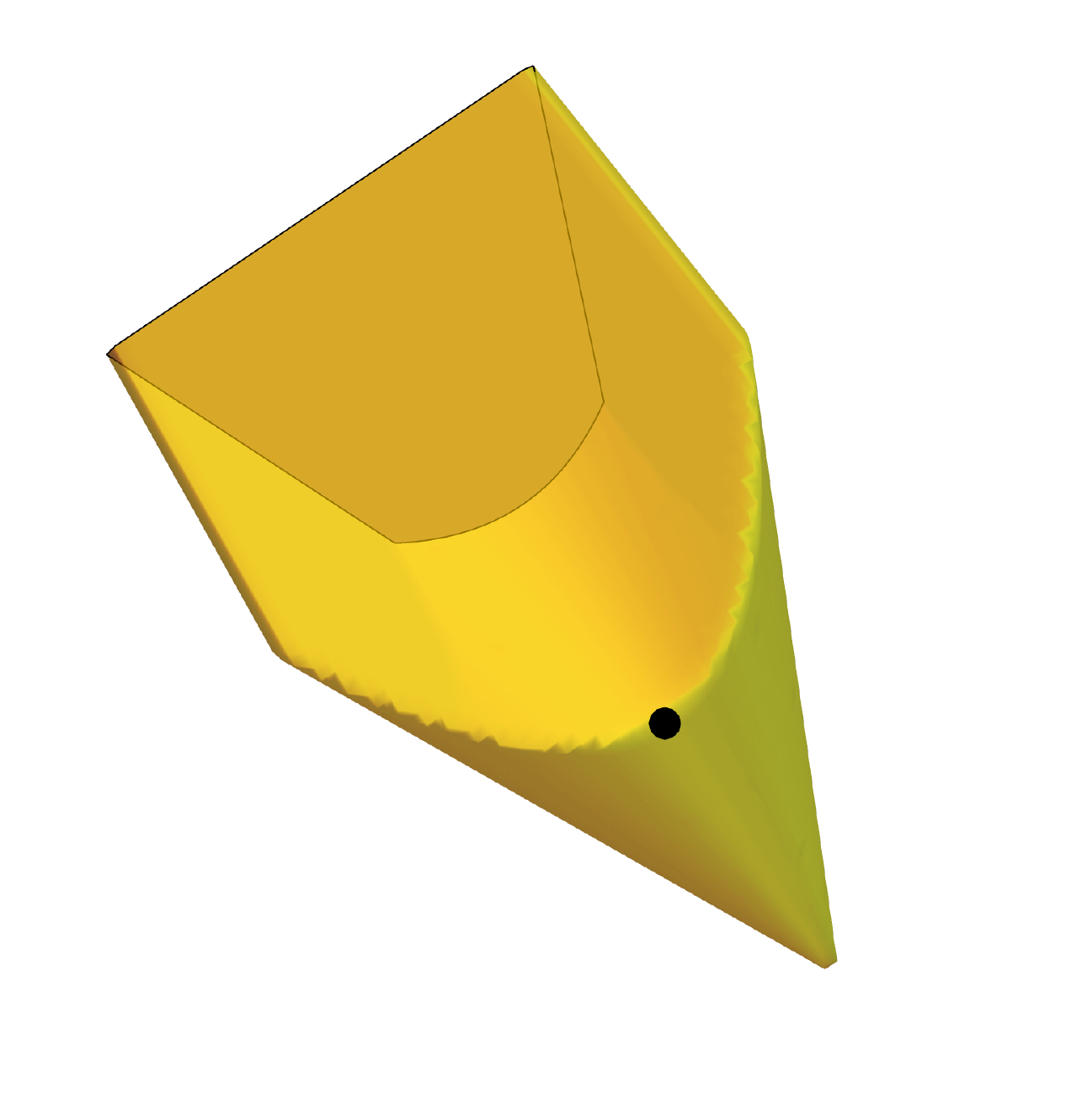}
       \caption{Each $(a,b,c)$ in the above set indexes a 
    $2-$isospectral subgraph of $K_4$. There are $2-$sparsifers missing $1,2$ and $3$ edges indexed by the polyhedral faces of the set as explained in the text.\label{fig:K4new}}
    \end{figure}
 The geometric region describing admissible choices of weights has three flat sides coming from the intersections of $\mathcal{S}^2_+$, the cone of $2 \times 2$ positive semidefinite matrices, with the hyperplanes $a=4,$ which corresponds to the `lid' of the boat and $-a + \sqrt{2} c = 4, $ and $-a - \sqrt{2}c = 4$ which correspond to the two sides of the boat. The fourth 
    hyperplane $a-2b=4$ supports $\textup{Sp}_{K_4}(2)$ at $p=(4,0,0)$; the black dot in the top and front of 
    the bow of the boat.   
 In the relative interior of the flat side cut out by $a=4$ we get $2-$sparsifiers missing the edge $(1,2)$. In the relative interior of each of the flat sides cut out by $-a \pm \sqrt{2}c = 4$ we get $2-$sparsifiers missing two edges: either $(1,3)$ and $(2,3)$, or $(1,4)$ and $(2,4)$. On the intersections of $a=4$ and $-a \pm \sqrt{2} c = 4$ 
    we get $2-$sparsifiers that are spanning trees. The point $p$ is at the 
    intersection of $a=4$ and $a-2b=4$ and indexes a $2-$sparsifier missing the edges $(1,2)$ and $(3,4)$. 
   \begin{figure}[h]
        \centering
        \begin{tabular}{c c c}
         \begin{tikzpicture}
                  \tikzstyle{bk}=[circle, fill = white,inner sep= 2 pt,draw]
          \tikzstyle{red}=[circle, fill = red,inner sep= 2 pt,draw = red]
\node (v1) at (0,0) [bk] {$1$};
\node (v2) at (0,2) [bk] {$2$};
\node (v3) at (2,0) [bk] {$3$};
\node (v4) at (2,2)  [bk] {$4$};
\draw[thick] (v4) --(v1) -- (v3) --(v2) -- (v4) ;
\draw[cyan, dashed, thick] (v1) -- (v2);
\draw[cyan, dashed, thick] (v3) -- (v4);
\end{tikzpicture}
&             \begin{tikzpicture}
                  \tikzstyle{bk}=[circle, fill = white,inner sep= 2 pt,draw]
          \tikzstyle{red}=[circle, fill = red,inner sep= 2 pt,draw = red]
\node (v1) at (0,0) [bk] {$1$};
\node (v2) at (0,2) [bk] {$2$};
\node (v3) at (2,0) [bk] {$3$};
\node (v4) at (2,2)  [bk] {$4$};
\draw[thick] (v2) --(v4) --(v1); 
\draw[line width = 1 mm] (v3) --(v4); 
\draw[cyan, dashed,thick] (v1) -- (v2) --(v3) -- (v1);
\end{tikzpicture}   
&  \begin{tikzpicture}
                  \tikzstyle{bk}=[circle, fill = white,inner sep= 2 pt,draw]
          \tikzstyle{red}=[circle, fill = red,inner sep= 2 pt,draw = red]
\node (v1) at (0,0) [bk] {$1$};
\node (v2) at (0,2) [bk] {$2$};
\node (v3) at (2,0) [bk] {$3$};
\node (v4) at (2,2)  [bk] {$4$};
\draw[thick] (v2) --(v4) --(v1); 
\draw[thick] (v3) --(v4); 
\draw[cyan, dashed,thick] (v1) -- (v2) --(v3) -- (v1);
\end{tikzpicture}  \\
\(\begin{bmatrix} 
4 & 0 & -2 & -2 \\
- & 4 & -2 & -2 \\
- & - & 4 & 0 \\
- & - & - & 4
\end{bmatrix}
\)
& 
\(\begin{bmatrix}
4 & 0 & 0 & -4 \\ 
- & 4 & 0 & -4 \\
- & - & 50 & -50 \\ 
- & - & - & 58 
\end{bmatrix} 
\)
&
\(\begin{bmatrix}
4 & 0 & 0 & -4 \\
- & 4 & 0 & -4 \\ 
- & - & 4 & -4 \\
- & - & - & 12 
\end{bmatrix} \)
        \end{tabular}
        \caption{The first graph is the $2-$sparsifier of $K_4$ corresponding to  $p=(4,0,0)$. 
        The next two graphs correspond to  $(4, 100, 4\sqrt{2})$ and $(4,8,4\sqrt{2})$ which lie on the edge created by the intersection of $a=4$ and $-a+2\sqrt{c}=4$. They are both spanning trees with different weights. Dashed lines represent missing edges.}
        \label{fig:K4 2-sparsifiers}
    \end{figure}
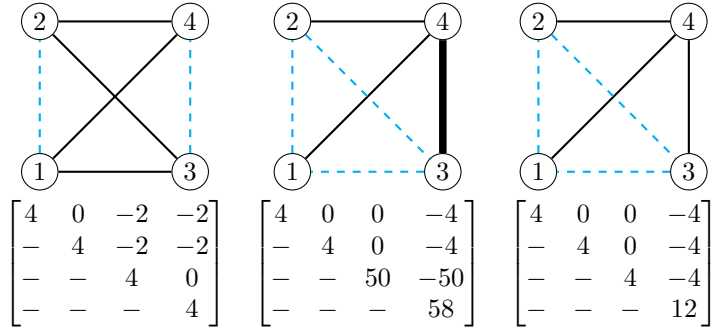
We conclude by pointing out various features of these two examples that will be made rigorous in subsequent sections.
\begin{enumerate}
    \item The set $\textup{Sp}_G(k)$ is the set of Laplacians of the $k-$isospectral subgraphs of $G$. These Laplacians are parameterized by $$\ell \leq {n-k+1 \choose 2}$$ variables corresponding to the degrees of freedom in a symmetric matrix of size ${n-k}$, allowing us to identify $\Sp_G(k)$ with a subset of $\RR^\ell$.
    \item The set $\textup{Sp}_G(k)$ is the intersection of the cone of positive semidefinite (psd) matrices, an affine space determined by the edges not in $G$, and a polyhedron described by linear inequalities corresponding to the weights of edges in $G$. In the first example, the psd constraints are automatically satisfied by all points that satisfy the linear constraints. 
    In the second example, the psd constraints contribute, making $\textup{Sp}_G(2)$ non-polyhedral. 
    \item The $k-$sparsifiers of $G$ are indexed by the faces of the polyhedron that survive in $\textup{Sp}_G(2)$. A $k-$sparsifier on a face of the polyhedron cut out by $t$ linear inequalities is missing at least $t$ edges. 
\end{enumerate}

\section{The Geometry of $k-$Sparsifiers} 
\label{sec:structure theorem}

We now describe the  set $\Sp_G(k)$ of 
$k-$isospectral subgraphs of $G$.  The main structure statement is in \S~\ref{subsec:structure theorem}, followed by a discussion of the implied geometry 
in \S~\ref{subsec:geometry of sparsifiers}. This includes a precise statement about where the sparsifiers are located in $\Sp_G(k)$. In \S~\ref{subsec:example K5} we 
illustrate the computation of $\Sp_G(k)$ on a full example. 

\subsection{Main Structure Theorem}\label{subsec:structure theorem}
As explained in Lemma~\ref{lem: cant disconnect} every $k-$isospectral subgraph $\tilde{G}$ of $G=([n],E,w)$  is spanning and connected when $k \geq 2$. Since the Laplacians uniquely identify the subgraphs, we begin with a complete description of $\textup{Sp}_G(k)$, the set of Laplacians associated to all $k-$isospectral subgraphs of $G$.  

\begin{theorem} \label{thm:k sparsifier}
Let $G=([n],E,w)$ be a connected, weighted graph with eigenpairs 
    $(0,\phi_1), (\lambda_2, \phi_2), \ldots, (\lambda_n, \phi_n)$ where $0= \lambda_1 < \lambda_2 \leq  \cdots \leq \lambda_n$ and 
    $\{\phi_i\}_{i=1}^n$ are orthonormal. Fix $2 \leq  k \leq n$ and define the matrices 
    $$\Phi_k = \begin{bmatrix} \phi_1 & \cdots & \phi_k \end{bmatrix} \in \RR^{n \times k}, \,\,\,
    \Phi_{>k} = \begin{bmatrix} \phi_{k+1} & \cdots & \phi_n \end{bmatrix} \in \RR^{n \times (n-k)},$$
    $$\Lambda_k = \Diag(0,\lambda_2, \ldots, \lambda_k) \in \RR^{k \times k}. $$ Then 
    the set of  
    Laplacians of all $k-$isospectral subgraphs of $G$ is  
    \begin{align*}
    \textup{Sp}_G(k) = 
    \left\{ L = \underbrace{\Phi_k \Lambda_k \Phi_k^\top + \lambda_k \Phi_{>k} \Phi_{>k}^\top}_{F} + \Phi_{>k} Y \Phi_{>k}^\top \,:\, 
    \begin{array}{cc}
         Y \in \cal S_+^{n-k}  \\
         L_{st} \leq 0 \,\,\forall \,\, (s,t) \in E \\
         L_{st} = 0 \,\, \forall \,\, s \neq t, \,(s,t) \not \in E
    \end{array}\right\}.
    \end{align*} 
\end{theorem}
    \emph{\nameref{pf: k-sparsifier}}

For $G$ and $k$ fixed, the matrix $F := \Phi_k \Lambda_k \Phi_k^\top + \lambda_k \Phi_{>k} \Phi_{>k}^\top$ is fully determined and easily computable. The set $\Sp_G(k)$  lies in the psd cone $\mathcal{S}^n_+$ and each $k-$isospectral subgraph  of $G$ has a Laplacian of the form $L = F + \Phi_{>k} Y \Phi_{>k}^\top$ where $Y$ is a psd matrix in $\mathcal{S}^{n-k}_+$.  The entries of 
$Y$ satisfy linear inequalities indexed by the edges present in $G$ and linear equations indexed by the edges missing in $G$. The smaller the value of $k$, the larger the potential dimension of $\Sp_G(k)$.

\subsection{The Geometric Structure of $\Sp_G(k)$} \label{subsec:geometry of sparsifiers}
Each $Y$ satisfying the conditions in the description of $\Sp_G(k)$ determines a matrix $L$ which in turn determines a unique  $k-$isospectral subgraph $\tilde{G}$ of $G$ with $L_{\tilde{G}} = L$.
Identifying  $Y$ with its $n-k+1 \choose 2$ entries $y_{ij}$ in the upper triangle (including the diagonal), 
we may  take $\Sp_G(k)$ to be a subset of 
$\RR^{n-k+1 \choose 2}.$ Thus, $\Sp_G(k)$ is the 
set of $k-$isospectral subgraphs of $G$, or the set of Laplacians of $k-$isospectral subgraphs of $G$, or a subset of $\RR^{n-k+1 \choose 2}$, each point of which provides a $Y$ that in turn provides the Laplacian 
$L = F + \Phi_{>k} Y \Phi_{>k}^\top$ of a $k-$isospectral subgraph of $G$. This last interpretation is the most helpful for computations and we describe its geometry.

The expressions $L_{st}$ are linear functions in the entries $y_{ij}$ of $Y$.
If we denote the columns of $\Phi_{>k}^\top$ as $c_1, c_2,  \ldots, c_n \in \RR^{n-k}$ then the 
$st$ entry of $\Phi_{>k} Y \Phi_{>k}^\top$ is 
\begin{align}
c_s^\top Y c_t = \langle Y, c_s c_t^\top \rangle. 
\end{align}
Therefore, 
\begin{align}
L_{st} = F_{st} + \langle Y, c_s c_t^\top \rangle.
\end{align}
Define the polyhedron 
\begin{align} \label{eq:polyhedron}
P_G(k) := & 
    \left\{ (y_{ij}) \in \RR^{n-k+1 \choose 2} \,:\, 
    \begin{array}{cc}
         L_{st} \leq 0 \,\,\forall \,\, (s,t) \in E
    \end{array} 
    \right\}
\end{align}
and the {\em spectrahedron} (the intersection of a psd cone with an affine space)

\begin{align} \label{eq:spectrahedron} 
\begin{split}
S_G(k) := & 
    \left\{ (y_{ij}) \in \RR^{n-k+1 \choose 2} \,:\, 
    \begin{array}{cc}
         L_{st} = 0 \,\, \forall \,\, (s,t) \not \in E, s \neq t,\\
         Y \succeq 0
    \end{array} 
    \right\}
    \end{split}
\end{align}

The set $\textup{Sp}_G(k)$ is the intersection of the (convex) polyhedron $P_G(k)$ and the (convex) spectrahedron $S_G(k)$. Alternately it is the 
intersection of the psd cone $\mathcal{S}^{n-k}_+$ with the (convex) polyhedron defined the linear inequalities $L_{st} \leq 0$ for all $(s,t) \in E$ and the linear equations $L_{st} = 0$ for all $(s,t) \not \in E, s \neq t$. Further, 
the sets $\Sp_G(k)$, when thought of as sets of Laplacians,  are nested since any subgraph that shares the first $k$ eigenvalues and eigenvectors with $G$ also shares the first $k-1$ of them with $G$. This aligns with our intuition 
that the difficulty of finding a $k-$sparisfier must increase with $k$.

\begin{corollary} \label{cor:nestedness}
The sets $\{\textup{Sp}_G(k)\}$ are convex and nested;  $\Sp_{G}(k) \subseteq \Sp_G(k-1)$ for all 
$2 \leq k \leq n$.
\end{corollary}

In Example~\ref{ex:butterfly}, the psd constraints on $Y$ are redundant while in Example~\ref{ex:K4}, the psd constraints are active. In both examples, the polyhedron $P_{G}(k)$ is unbounded.

Recall that a $k-$sparsifier of $G$ is a $k-$isospectral subgraph of $G$ that has at least one less edge than $G$.
A natural goal when sparsifying a graph is to delete as many edges as possible (while possibly reweighting others to preserve the first $k$ eigenvectors and eigenvalues). 
The equations $L_{st}=0$ for all $(s,t) \not \in E$ ensure that $E(\tilde{G}) \subseteq E(G)$ for the graphs $\tilde{G} \in \Sp_G(k)$.
An edge $(s,t) \in E$ is missing in $\tilde{G}$ if and only if the linear inequality $L_{st} \leq 0$ holds at equality 
in $L_{\tilde{G}}$. Using this, we can locate the $k-$sparsifiers in $\Sp_G(k)$. 

\begin{corollary}\label{cor:23}
Suppose a $k-$sparsifier $\tilde{G}$ of $G$ has Laplacian 
$$L = F + \Phi_{>k} Y_{\tilde{G}} \Phi_{>k}^\top$$
in ${\textup{Sp}}_G(k)$.
The sparsity of  $\tilde{G}$, namely the edges of $G$ that are not present in $\tilde{G}$, is  determined by the face of $P_G(k)$ that contains $Y_{\tilde{G}}$. More precisely, $(s,t) \in E \setminus \tilde{E}$ if and only if the inequality  $L_{st} \leq 0$  holds at equality on the face containing $Y_{\tilde{G}}$. 
\end{corollary}
\emph{\nameref{pf: 23}}

Less formally, the sparsity patterns of $k-$sparsifiers of $G$ are indexed by the faces of $P_G(k)$ contained (at least partially) in $\textup{Sp}_G(k)$.
If $Y_{\tilde{G}}$ lies on an $\ell$-dimensional face of $P_G(k)$, then $\tilde{G}$ is missing at least $|E|-\ell$ edges that were present in $G$. However, it can miss many more edges in particular graphs, because the equations defining the off-diagonal entries of $L_{\tilde{G}}$ need not be unique, see Example~\ref{subsec:example K5}.

In general, $\Sp_G(k)$ depends on the choice of eigenbasis of $L_G$ since the construction picks $k$ specific eigenpairs to freeze in the $k-$isospectral subgraphs of $G$. Therefore, the $k-$sparsifiers you get will depend on the choice of eigenbasis. See for example the two different $\Sp_{K_5}(4)$ in \S~\ref{subsec:example K5}. However, 
this discrepancy only occurs when we freeze some eigenpairs with a specific eigenvalue and leave out others. In other words, if the choice of $k$ preserves entire eigenspaces, then $\Sp_G(k)$ does not depend on the choice of eigenbasis of $L_G$.

  \begin{theorem}\label{thm: Sp2 independent of basis}
     If $\l_k < \l_{k+1}$, then $\Sp_G(k)$ is independent of the choice of eigenbasis of the Laplacian $L_G$.
  \end{theorem}
  \emph{\nameref{pf: Sp2 independent of basis}}

\subsection{The complete graph $K_5$}
 \label{subsec:example K5}
 We now compute $\Sp_G(k)$ in detail for two values of $k$ for the complete graph $G = K_5$. Suppose we fix the following eigenpairs of $G$:
 
    \begin{tiny}
    $$  \left(0, \phi_1 = \frac{1}{\sqrt{5}} \begin{bmatrix} 1\\1\\1\\1\\1\end{bmatrix}\right), 
    \left(5, \phi_2 = \frac{1}{\sqrt{2}}\begin{bmatrix} 1 \\ -1 \\ 0 \\ 0\\ 0\end{bmatrix}, 
    \phi_3 = \frac{1}{\sqrt{2}}\begin{bmatrix} 0\\0\\1\\-1\\0\end{bmatrix}, 
    \phi_4 = \frac{1}{2}\begin{bmatrix} 1\\1\\-1\\-1\\0\end{bmatrix}, 
    \phi_5 = \frac{1}{2\sqrt{5}}\begin{bmatrix} 1\\1\\1\\1\\-4\end{bmatrix}\right).$$
    \end{tiny}

    For any $k$, the first step in computing $\Sp_G(k)$ is to compute the fixed matrix $F$. Here is a helpful observation that holds in general.
    
\begin{lemma} \label{lem: F for one fixed eigenval}
For any graph $G$, if $\l_2 = \cdots = \l_k$, then $F$ is a scaling of the Laplacian of $K_n$, namely $F = \frac{\l_2}{n}L_{K_n} = 
\l_2 (I_{n} -  \frac{1}{n}J_n) $, where $J_n$ is the $n\times n$ matrix of all ones.  
\end{lemma}
\emph{\nameref{pf: F for one fixed eigenval}}

Consider $k=3$. 
 By Lemma~\ref{lem: F for one fixed eigenval}, $F = L_{K_5} = 5I_5- J_5$. Setting \[ Y = \begin{bmatrix}
            a & c \\ c & b 
        \end{bmatrix}, 
        \]  
          compute $L = F + \Phi_{>3}Y \Phi_{>3}^\top$. 
          Since $K_5$ has no missing edges, all off-diagonal entries of $L$ contribute an inequality $L_{st} \leq 0$ to the polyhedron $P_{K_5}(3)$. We list the inequalities with the edges contributing each inequality on the right:
         \begin{align}
             5a + b + 2 \sqrt{5} c &\leq 20 && (1,2)\\
             5a + b - 2 \sqrt{5} c &\leq 20 && (3,4)\\
             -5a  + b  \phantom{, - 2 \sqrt{5} c} &\leq 20  &&(1,3),(1,4),(2,3),(2,4)\\
                  -b  - \sqrt{5}c  &\leq 5   && (1,5),(2,5)\\
                 -b    + \sqrt{5}c &\leq 5   && (3,5),(4,5)
        \end{align}
         Note that several edges can contribute the same inequality. 
The polyhedron $P_{K_5}(3)$ is a square pyramid as seen in Figure~\ref{fig:K5poly}. The base is given by the hyperplane $-5a+b=20$. Intersecting it with $\mathcal{S}^2_+$ we get $\Sp_3(K_5)$, seen in Figure~\ref{fig:K5k3sparsifiers}. 

\begin{figure}[h]
    \centering
    \includegraphics[scale=0.3]{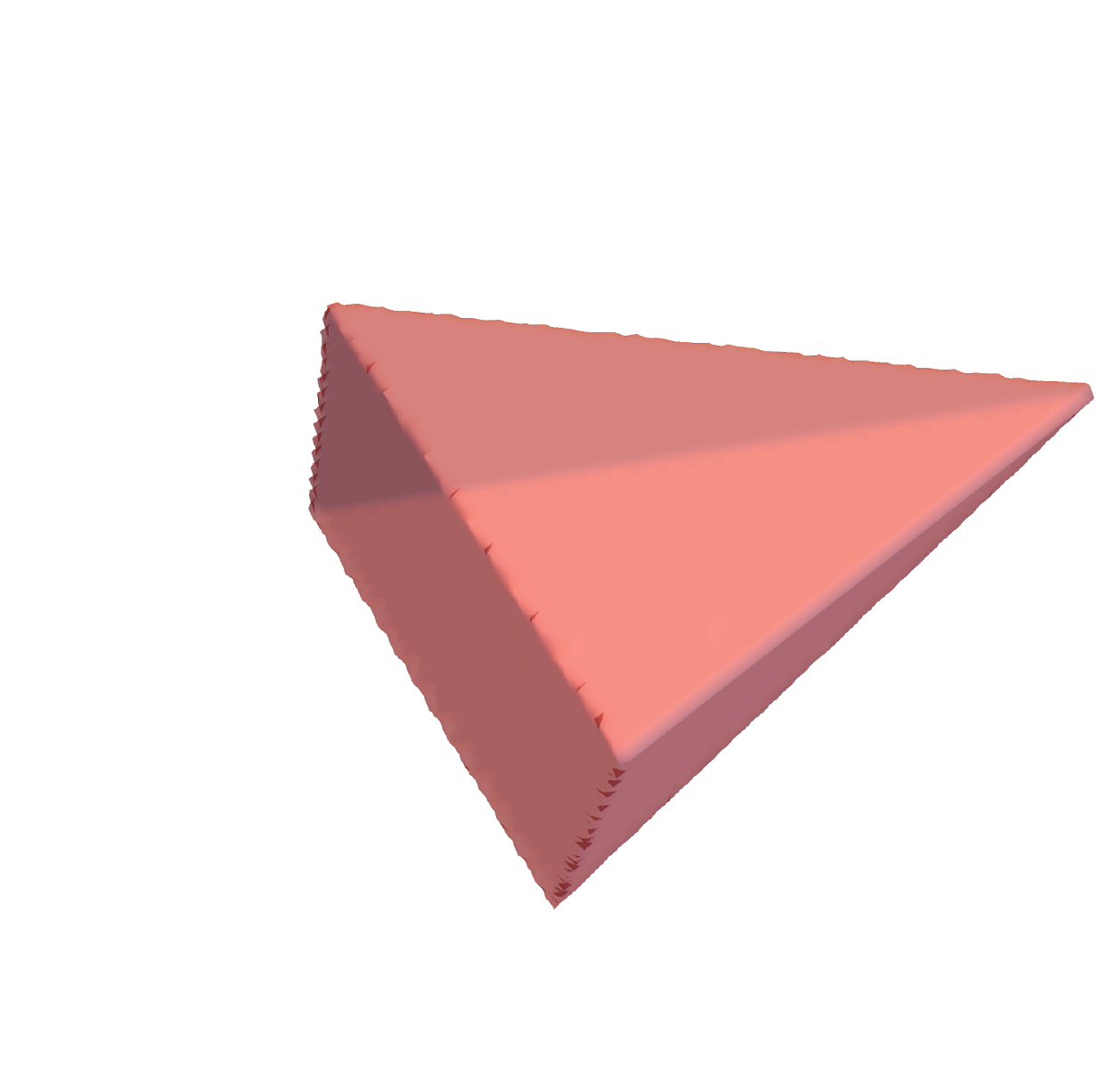}
    \quad\quad\quad
    \includegraphics[scale=0.3]{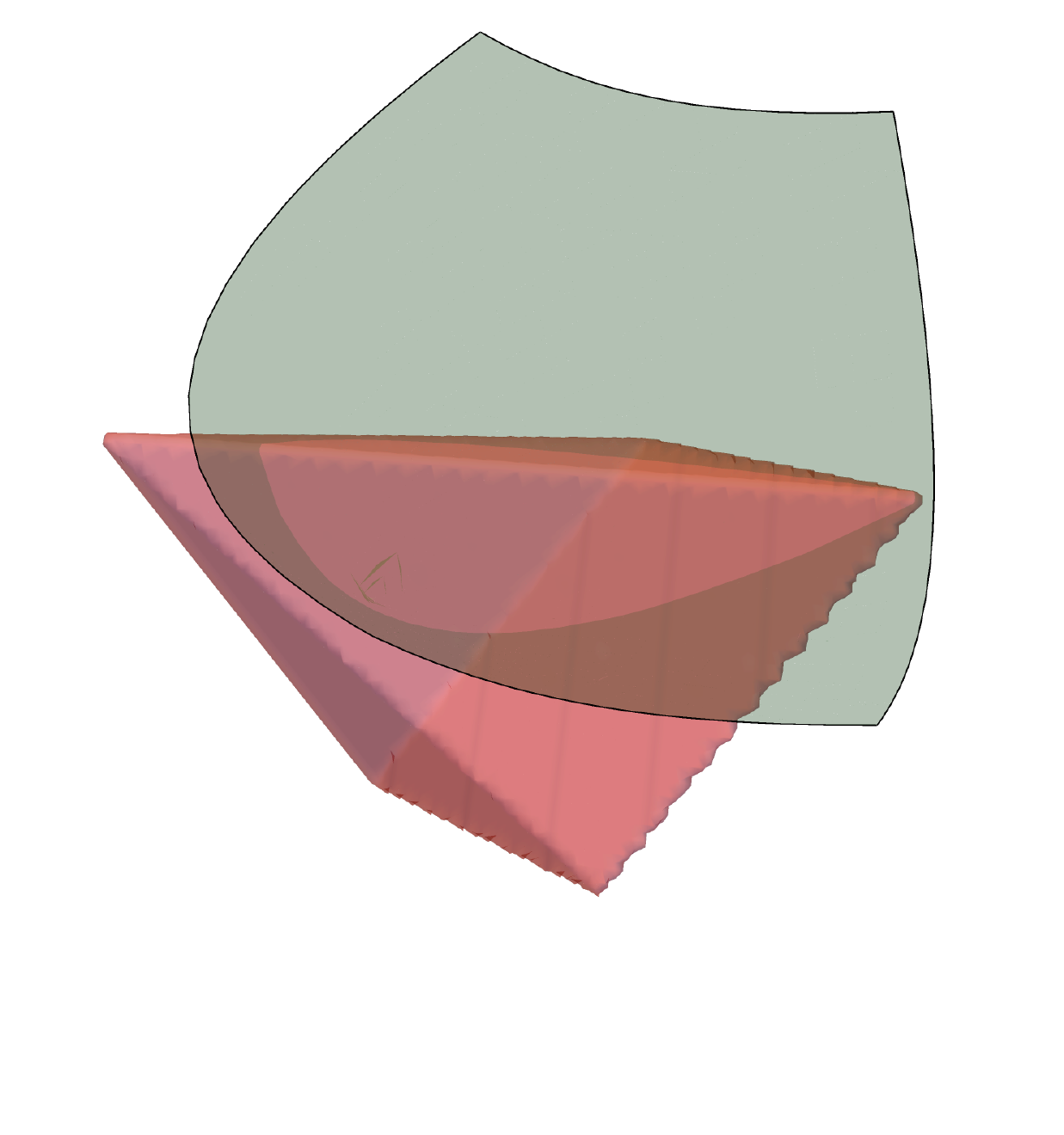}
    \caption{The polyhedron $P_{K_5}(3)$ and its intersection with $\mathcal{S}^2_+$.}
    \label{fig:K5poly}
\end{figure}

The two flat faces of $\textup{Sp}_{K_5}(3)$ are defined by  $5a + b \pm 2 \sqrt{5} c = 20 $. The interiors of these faces index $3-$sparsifiers missing edge $(1,2)$ or $(3,4)$.
 \begin{figure}[h]
     \centering
    \includegraphics[scale=0.4]{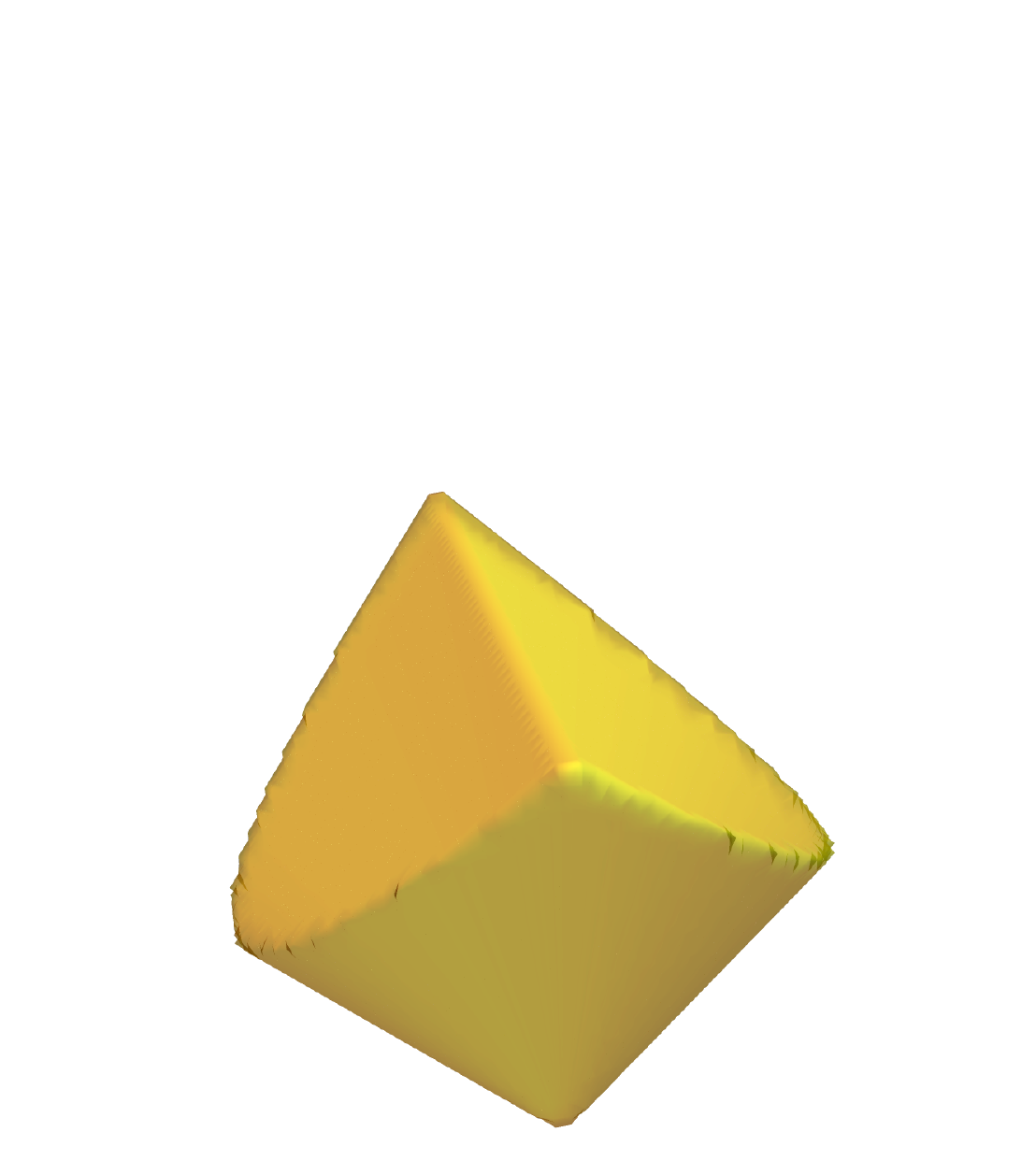}
     \includegraphics[scale=0.3]{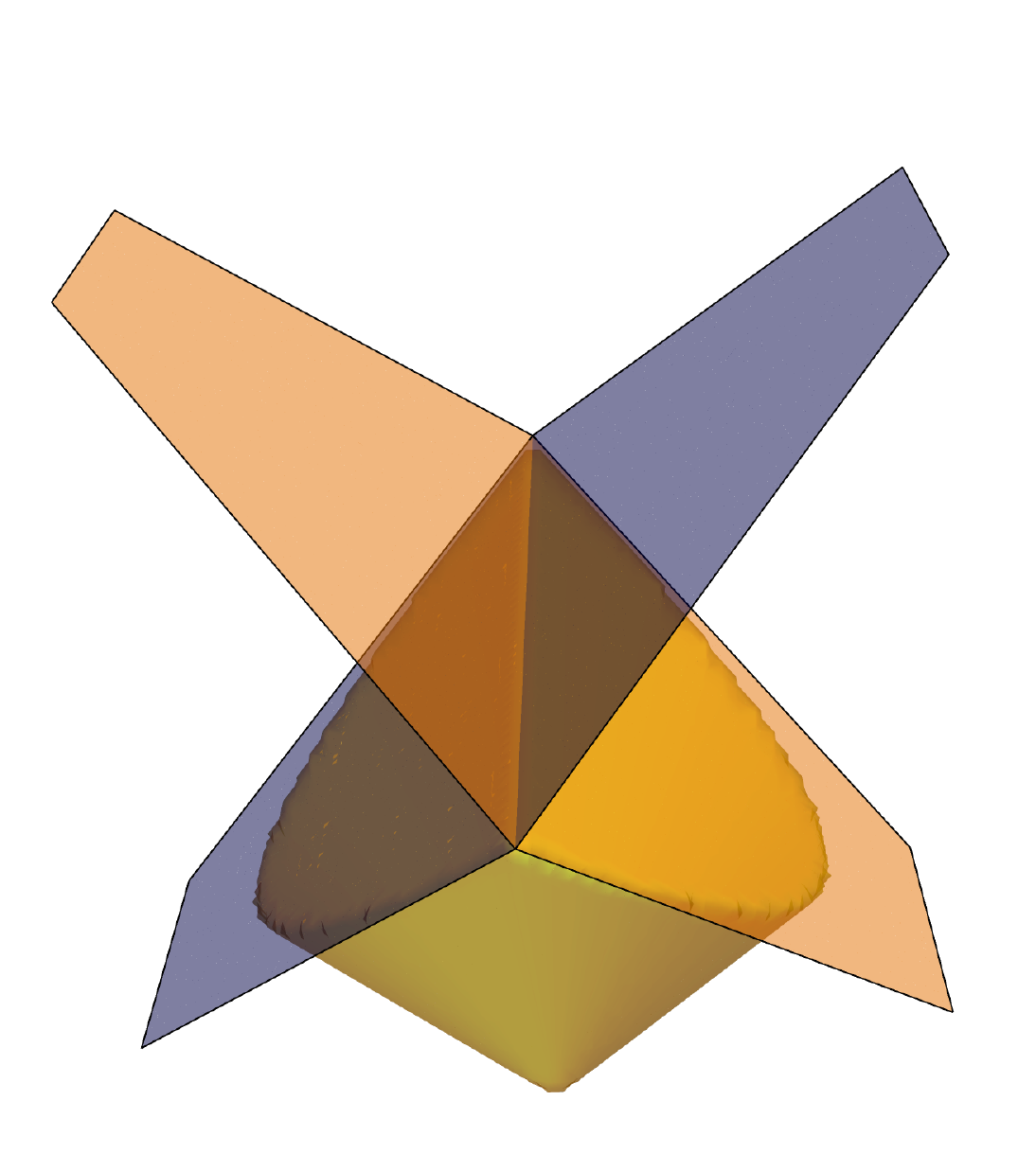} \quad\quad 
        \includegraphics[scale=0.3]{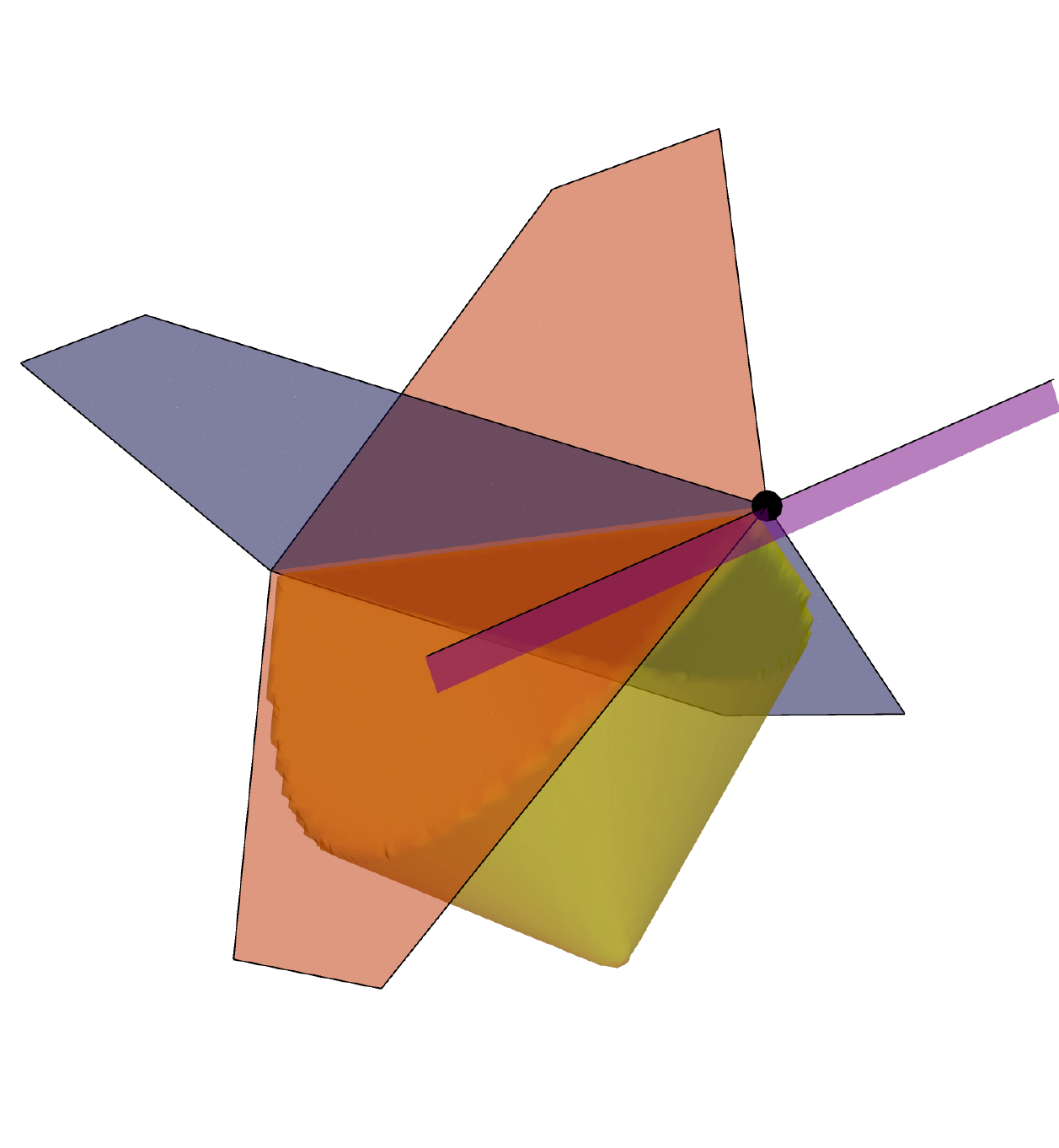}
     \caption{The first picture shows $\Sp_{K_5}(3)$, the set of points  $(a,b,c)$ indexing $3-$isospectral subgraphs of $K_5$. Its flat sides are cut out by the first two hyperplanes as seen in the middle picture. 
     The last picture shows the vertex $(0,20,0)$ of $\Sp_{K_5}(3)$ (black dot) at which the first three hyperplanes intersect.}
     \label{fig:K5k3sparsifiers}
 \end{figure}

 The hyperplanes with equations $-b  \pm \sqrt{5}c= 5$ do not touch $\textup{Sp}_{K_5}(3)$.  This means that with this choice of basis, $K_5$ has no $3-$isospectral subgraph that 
 is missing any edge incident to vertex 5.  The intersection of the other three hyperplanes with $\textup{Sp}_{K_5}(3)$ is the point $(0 , 20, 0)$.  This point produces a spanning tree sparsifier consisting of all edges incident to vertex 5. We depict these intersections in Figure~\ref{fig:K5k3sparsifiers}.

 \bigskip 
        Now consider $k=4$. 
        If we hold the first four eigenpairs fixed,  then $\Phi_{> 4} = \phi_5$,  $Y =[a]$, and 
               the inequalities defining $P_{K_5}(4)$ are
        \begin{align}
           -5 \leq a \leq 20.        
        \end{align}
        The lower bound comes from the edges $(1,5),(2,5),(3,5),(4,5)$ and the upper bound comes from the other edges. Since $Y \succeq 0$, $a \geq 0$,  and  $\textup{Sp}_{K_5}(4) = [0,20]$. This means that there is a unique $4-$sparsifier indexed by $a=20$ and it corresponds to the spanning tree with edges $(1,5),(2,5),(3,5),(4,5)$.

        On the other hand, if we had chosen $\Phi_{>4} = \phi_4$, then $P_{K_5}(4)$ is defined by $-4 \leq a \leq 4 $, 
        which makes $\textup{Sp}_{K_5}(4) = [0,4]$. The bound $a \leq 4$ is given by the edges $(1,2)$ and $(3,4)$. Therefore, the unique $4-$sparsifier is now indexed by  $a=4$ and is missing edges $(1,2)$ and $(3,4)$. In particular, there is no spanning tree $4-$sparsifier for this choice of eigenvectors to hold constant.
        
        Lastly, if we had chosen $\Phi_{> 4} = \phi_2$, then 
        $\Sp_{K_4}(4) = [0, \infty] = \mathcal{S}^1_+$. The only inequality is of the form $a \geq -2$ and the hyperplane $a = -2$ does not support $\Sp_{K_4}(4)$. Therefore, there are no $4-$sparsifiers for this ordering of eigenvectors.

        By Theorem~\ref{thm: Sp2 independent of basis} we see this wild behavior as the choice of $\phi_n$ changes because we are not choosing to preserve all the eigenvectors corresponding to $\l_2 = 5$. 
        
\section{Bounds on the Maximal Sparsification} \label{sec:lin}

Given our notion of a $k-$sparsifier, the main question now is: how large can we choose $k$ and still obtain  nontrivial sparsification? Or, conversely, if we wish to aggressively remove edges, how much of the spectrum can we reasonably hope to preserve?  The purpose of this section is to introduce some basic bounds. At the heart of it 
is a \textit{Linear Algebra Heuristic}, introduced in \S \ref{sec: lin alg bound}, which provides a guideline for what is true generically. 
Before discussing the general case, we start with two different extreme settings: the (weighted) complete graph can be sparsified to a high degree and a tree can never be sparsified.

\begin{theorem} \label{thm:Kn sparsifies for all k}
    If $G$ is a weighted complete graph, then $G$ is guaranteed to have a $k-$sparsifier for all $k \leq n-2$.   \label{lem: weighted complete sparsifies}
\end{theorem}
\emph{\nameref{pf: weighted complete sparsifies} }

This result is sharp.
As we saw at the end of Example~\ref{subsec:example K5},
if $\phi_n$ is chosen to be $(1/\sqrt{2}, -1/\sqrt{2},0,0,0)$ then $K_5$ has no $4-$sparsifiers. This counterexample is also sharp in the sense that whenever the eigenvector $\phi_n$ has at least three nonzero coordinates, then the graph has an $(n-1)-$sparsifier -- see Proposition~\ref{prop:n-1 sparsifier}.\\

On the other extreme, if $G$ is a tree, then removing any edge disconnects the graph. Since all $k-$isospectral subgraphs of $G$ are connected when $k \geq 2$ (Lemma~\ref{lem: cant disconnect}), we have the following. 

\begin{lemma}
    If $G = ([n], E, w)$ is a tree, then even at $k=2$, a $k-$isospectral subgraph of $G$ cannot lose any edges.
\end{lemma}
More generally, recalling that the multiplicity of the first eigenvalue captures the number of connected components, we see that any sparsifier preserving the eigenspace corresponding to eigenvalue 0 needs to necessarily preserve the number of connected components.

\subsection{The Linear Algebra Heuristic} \label{sec: lin alg bound}
In our model,  $\Sp_k(G)$ is parameterized by at most $n - k+1 \choose 2$ variables corresponding to the distinct entries of $Y \in \cal S^{n-k}_+$. Every missing edge of $G$ contributes a linear equation to the description of $\Sp_G (k)$ while every edge contributes a linear inequality.  If $G$ is sufficiently generic, then we would expect that the equations are linearly independent; equivalently, each missing edge decreases the dimension of $\Sp_G(k)$ by one. However $G$ itself is in $\Sp_G(k)$ for all $k$. So, if $G$ is missing at least $\binom{n -k+1}{ 2}$ edges, then we expect $\Sp_G(k)$ to contain just $G$, and $G$ to have no $k-$sparsifiers. 
We refer to this basic principle as the Linear Algebra Heuristic, which tells us what one can generically expect of $k-$sparsifiers.

\begin{principle}[Linear Algebra Heuristic] \label{lem: generic lin alg bound}
    If $G=([n],E,w)$ is a `generic' graph and 
    $$|E| \leq {n \choose 2} - {n-k+1 \choose 2}$$
    then, generically, 
    the only $k-$isospectral subgraph of $G$ is $G$ itself. 
\end{principle}

The notion of `generic' is to be understood as follows: `typical' linear systems are solvable as long as the number of variables is at least as large as the number of equations. This is not always true: the equations could be linearly dependent and only be solvable for particular right-hand sides. However, having such a dependence is delicate and not `generically' the case: in particular, if one were to perturb the weights $w$ of a graph ever so slightly, one would expect the Linear Algebra Heuristic to apply to `most' (in the measure-theoretic sense) perturbations. Simultaneously, since having such dependencies is non-generic, failures of the Linear Algebra Heuristic are interesting and a sign of a great degree of underlying structure.\\

We quickly mention an easy sample application of the Linear Algebra Heuristic. If we consider all unweighted  graphs with $n=12$ vertices and $|E| = 36$ edges, then 
according to the Linear Algebra Heuristic, we expect that `generically', there are $4-$sparsifiers, but the only $5-$isospectral subgraph is the graph itself.
Some numerical experiments show that this is typically true for Erd\H{o}s-Renyi random graphs $G(12, 1/2)$ conditioned on having 36 edges.

\subsection{Exceptions}
To be clear, no direction of the Linear Algebra Heuristic is true for all graphs.  
Figure~\ref{fig: beatsLinAlg} exhibits an unweighted graph $G$ with $n=12$ and 
$$|E|=36 \leq {n \choose 2} - {n-5+1 \choose 2} = 38$$
that sparsifies up to $k=8$. Unsurprisingly, these exceptions typically possess some type of symmetry.

\begin{center}
    \begin{figure}[h!]     
{\includegraphics[width=0.5\textwidth]{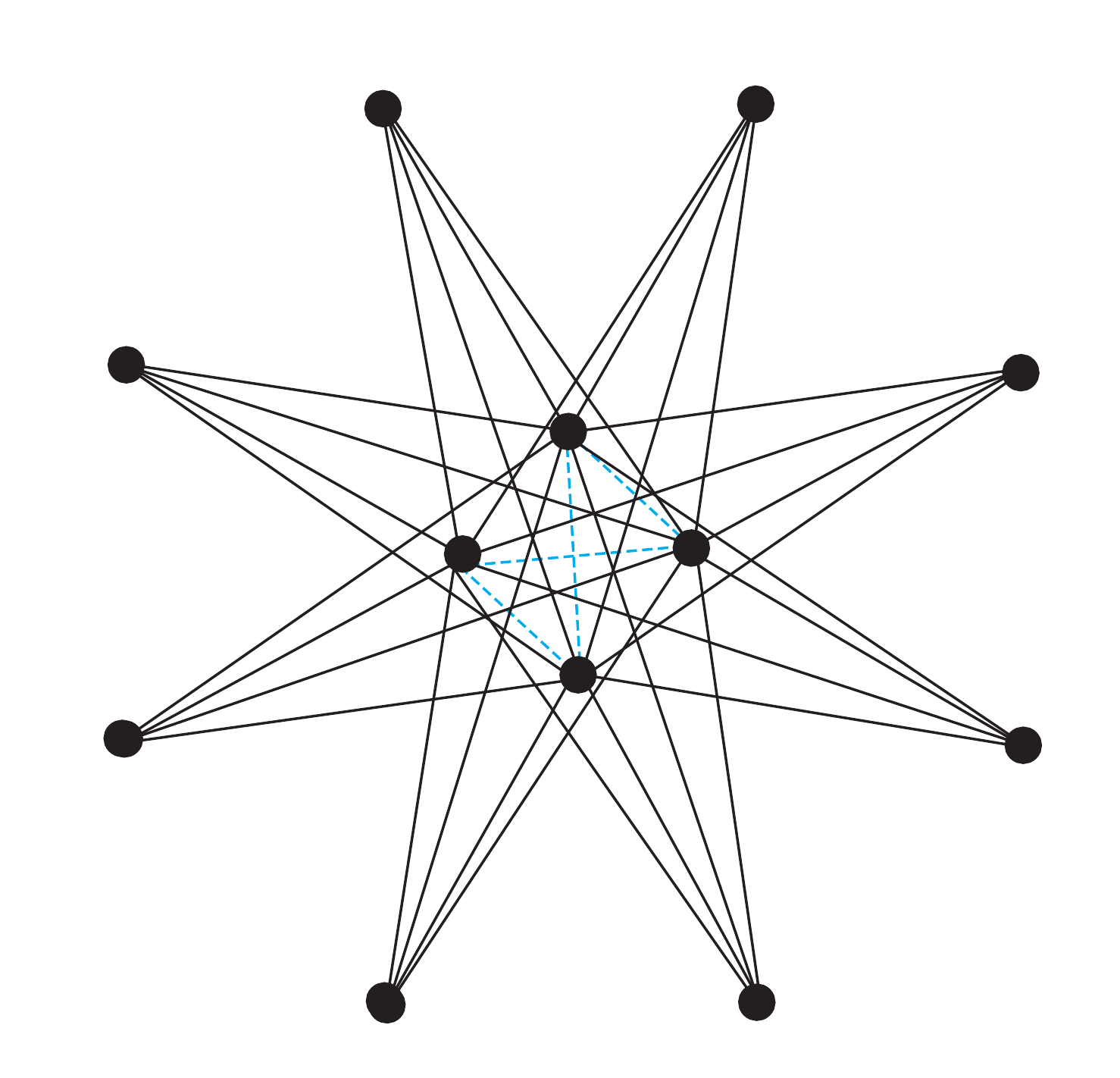}}
    \caption{A graph on $n=12$ vertices and $|E|=36$ edges that admits an $8-$sparsifier with 4 edges being deleted. The Linear Algebra Heuristic predicts that this should not happen} 
    \label{fig: beatsLinAlg}
        \end{figure}
\end{center}

It is also possible for a graph to have more edges than the difference of binomial coefficients in the linear algebra bound 
and not have $k-$sparsifiers.
Theorem~\ref{thm: graphs with no n-2 sparsifiers} exhibits a family of weighted graphs on $n$ vertices where \[|E| > {n \choose 2} - {n-(n-2)+1 \choose 2} = {n \choose 2} - 3,\] and yet, $G$ has no $(n-2)-$sparsifiers. The reason is that the spectrahedron 
$S_G(n-2)$ lies strictly inside the polyhedron $P_G(n-2)$ and hence no face of $P_G(n-2)$ supports $\Sp_G(n-2)$, see Figure~\ref{fig: cant sparsify, Y region}.

\begin{theorem} \label{thm: graphs with no n-2 sparsifiers}
    For each $n \geq 4$, there is a weighted graph  missing only one edge that does not sparsify at $k = n-2$. 
\end{theorem}
\emph{\nameref{pf: graphs with no n-2 sparsifiers} }

 For $n=4$, Theorem~\ref{thm: graphs with no n-2 sparsifiers} provides an example of a non-tree weighted graph that does not sparsify even at $k=2$ (see Figure~\ref{fig: graph cant sparsify}).   In contrast, the unweighted graph missing exactly one edge always has a spanning tree $2-$sparsifier. 
 
 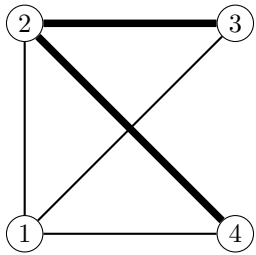
\begin{figure}[h!]
    \centering
 \begin{tikzpicture}[scale = .7]
        \tikzstyle{bk}=[circle, fill = white,inner sep= 2 pt,draw]
\node (v1) at (-2,-2) [bk] { 1};
\node (v2) at (-2,2) [bk] {2};
\node (v3) at (2,2) [bk] {3};
\node (v4) at (2,-2)  [bk] {4};
\draw[thick] (v2) -- (v1) -- (v3);
\draw[line width = 1 mm] (v3) -- (v2) -- (v4);
\draw[thick] (v1) -- (v4);
   \end{tikzpicture}
    \caption{A weighted graph that cannot sparsify even at $k=2$.  Thinner edges have weight 1, and thicker edges have weight 11.}
    \label{fig: graph cant sparsify}
\end{figure}

\begin{theorem} \label{thm:}
    Let $G = ([n],E)$ be the unweighted complete graph missing an edge.  Then $\Sp_G(2)$ contains a spanning tree sparsifier for any choice of eigenbasis.  \label{thm: unweighted complete minus one edge, SP2}
\end{theorem}
\emph{\nameref{pf: unweighted complete minus one edge, SP2} }

These last two theorems highlight the role of weights in the geometry of isospectral subgraphs and sparsifiers. The geometry is controlled by the arithmetic of the weights and the ensuing eigenstructure of the Laplacian, not just by combinatorics.

Lastly, we discuss a particularly interesting extremal case: the {\em cube graph} $Q_d$ on $V=\left\{0,1\right\}^d$ where any two vertices are connected if they differ on exactly one coordinate. The Linear Algebra Heuristic suggests that 
$$ d \cdot 2^{d-1} = |E|  \sim  {2^d \choose 2} - {2^d-k+1 \choose 2} $$
should be the natural cut-off after which no further sparsification is possible. Solving the quadratic polynomial, this predicts that as $d \rightarrow \infty$, we can sparsify up to $k \sim d/2$.  
Due to the extraordinary symmetry of the cube graph and the special structure of its eigenvalues and eigenvectors, sparsification up to a higher value of $k$ is possible, up to $k=d$. However, the first two eigenspaces uniquely determine the cube graph, which is to say that no sparsification is possible at $k = d+1$.

\begin{theorem} \label{thm:cube}
There is no sparsification of the cube graph $Q_d$ that preserves the first two eigenspaces (the first $d+1$ eigenvectors and eigenvalues) except for the trivial sparsification which leaves all edge weights invariant; $w_{ab} = 1$. 
\end{theorem}
\emph{\nameref{pf:cube}}

One would naturally wonder whether other graphs with special symmetries might perhaps admit a similar type of rigidity structure; this seems like an interesting problem but is outside the scope of this paper.\\

\section{Families of Band-limited Spectral Sparsifiers} \label{sec:families}
The purpose of this section is to discuss different graph families and describe what can be said about their
sparsification properties. This leads to problems at the interface of combinatorics, polyhedral geometry
and spectral geometry.

\subsection{The Complete Graph.} The unweighted complete graph plays a special role; it has eigenvalues $0$ and $n$ (with multiplicity $n-1$) with eigenspaces $\spanset\{\ones\}$ and $\ones^\perp$ respectively. Since there is only one non-trivial eigenspace, any sparsification of $K_n$ must depend on a choice of eigenbasis. Recall the 
$4-$sparsifiers of $K_5$ from \S~\ref{subsec:example K5}.

\begin{theorem} \label{thm: Kn spanning tree}
    There is a choice of eigenbasis so that for every $k < n$, a spanning tree is a $k-$sparsifier of $K_n$. The spanning tree  we exhibit is the star graph $K_{1,n-1}$, where every edge has equal weight $n$.
\end{theorem}
\emph{\nameref{pf: Kn spanning tree}}

   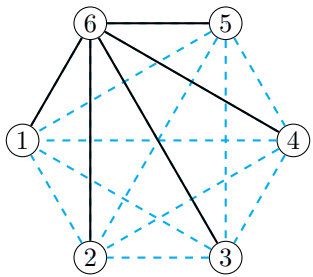
\begin{figure}[h]
        \centering
         \begin{tikzpicture}[scale=0.6]
  \tikzstyle{every node}=[circle,draw =black, fill=white,inner sep=1.5pt]
    \foreach \y[count=\a] in {1,2,...,6}
      {\pgfmathtruncatemacro{\kn}{60*\a+120}
       \node at (\kn:3) (\a){ \a} ;}
       \foreach \a in {2,...,6}
       {\foreach \b in {1, ..., \a}
   {\draw[dashed, thick, cyan] (\a)--(\b);} }
   \foreach \a in {1,...,5}
   {\draw[thick] (6)--(\a);} 
\end{tikzpicture}
        \caption{A $5-$sparsifier of $K_6$.}
        \label{fig:Kn sparsifier}
    \end{figure}

\subsection{The Wheel Graph}
Let $W_{n+1}$ denote the wheel graph on $n+1$ vertices, for $n \geq 3$.  We will assume that the vertices are ordered so that $n+1$ is the center of the wheel and that $i \sim (i+1 \mod n)$ for $ i \in [n]$ (see Figure~\ref{fig:wheel sparsifier}). 
 The spectrum of $W_{n+1}$ is well understood from its formulation as the join of the cycle $C_{n}$ and a single vertex $n+1$, see Table \ref{tab:Wheel spectra} for more details \cite{BrouwerHaemersSpectra, merris}.
The least nonzero eigenvalue of $L_{W_{n+1}}$ is $3 - 2 \cos (2\pi/n)$ which has multiplicity 2.

   \begin{figure}[h]
        \centering
         \begin{tikzpicture}[scale=0.6]
  \tikzstyle{every node}=[circle,draw =black, fill=white ,inner sep=1.5pt]
  \node at (0,0) (9){$9$};
    \foreach \y[count=\a] in {1,2,...,8}
      {\pgfmathtruncatemacro{\kn}{45*\a-45}
       \node at (\kn:3) (\a){\a} ;}
 \draw[dashed, thick, cyan] (1)--(2) -- (3) -- (4)--(5)--(6)--(7)--(8)--(1);
   \foreach \a in {1,...,8}
   {\draw[thick] (9)--(\a);} 
\end{tikzpicture}
        \caption{A $3-$sparsifier of $W_9$.}
        \label{fig:wheel sparsifier}
    \end{figure}
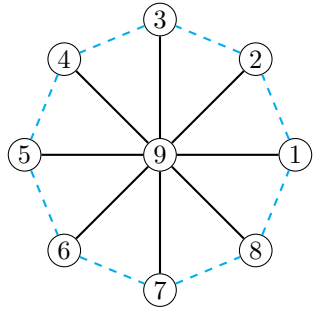

\begin{theorem} \label{thm:wheel}
     There is a $3-$sparsifier of $W_{n+1}$ which is a spanning tree.  The spanning tree we exhibit is the star graph $K_{1,n}$, where every edge has equal weight $ 3 - 2 \cos (2\pi/n) $.
\end{theorem}
\emph{\nameref{pf: wheel}}

We do not expect any $4-$sparsifiers of the wheel by the Linear Algebra Heuristic. Recalling that $W_{n+1}$ has $2n$ edges and $n+1$ vertices, this follows because \[|E| = 2n <  3(n-1) = \binom{n+1}{2} -  \binom{(n+1)-4 +1}{2}  .\]

\subsection{A General Principle} The existence of $k-$sparsifiers in a graph family is a statement about the eigenvectors and eigenvalues of the Graph Laplacian of graphs in the family. There appears to be a general principle that is worth recording.
\begin{theorem} \label{thm:examples}
    Let $G=([n],E, w)$ be a connected, weighted graph and let $T = ([n], E_T, w|_{E_T})$ be a spanning tree of $G$. Let $k \in [n]$ be arbitrary and let $\phi_1, \dots, \phi_k$ be eigenvectors corresponding to the $k$ smallest eigenvalues of the spanning tree $T$. Suppose that for all $(u,v) \in E$ either
    $$ (u,v) \in E_T \qquad \mbox{or} \qquad \phi_i(u) = \phi_i(v) \qquad \mbox{for all}~1 \leq i \leq k,$$ 
    then the spanning tree $T$ is a $k-$sparsifier of $G$ with respect to $\left\{ \phi_1, \dots, \phi_k \right\}.$
\end{theorem}
\emph{\nameref{pf: example}}
The result says that if a graph has the property that eigenvectors only change along the edges of a spanning tree, then that spanning tree is actually a $k-$sparsifier.

  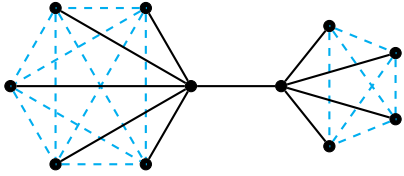
\begin{figure}[h]
        \centering
         \begin{tikzpicture}[scale=0.4]
  \tikzstyle{every node}=[circle,draw =black, fill=black ,inner sep=1.5pt]  
        \foreach \y[count=\a] in {1,2,...,6}
      {\pgfmathtruncatemacro{\kn}{60*\a-60}
       \node at (\kn:3) (\a){} ;}
       \foreach \a in {2,...,6}
       {\foreach \b in {\a,...,6}
   {\draw[dashed, thick, cyan] (\a)--(\b);} }
   \foreach \a in {2,...,6}
   {\draw[thick] (1)--(\a);} 
   \node at (6,0) (b1) {};
   \node at (7.6,2) (b2) {};
   \node at (9.8,1.1) (b3) {};
   \node at (9.8,-1.1) (b4) {};
   \node at (7.6,-2) (b5) {};
   \draw[thick] (1)--(b1); 
     \foreach \a in {2,...,5}
       {\foreach \b in {\a,...,5}
   {\draw[dashed, thick, cyan] (b\a)--(b\b);} }
    \foreach \a in {2,...,5}
   {\draw[thick] (b1)--(b\a);} 
\end{tikzpicture}
        \caption{A $2-$sparsifier of $B_{6,5}$.}
        \label{fig:barbell sparsifier}
    \end{figure}

The first application of Theorem \ref{thm:examples} is for the Barbell Graph $B_{n,n}$, which is two $n$-cliques joined by a bridge.  Its Laplacian is 
\[ L_{B_{n,n}} = 
\left(\begin{array}{@{}c|c@{}}
n  I_{n} - J_n & 0_{n,n} \\ \hline
  0_{n ,n} &  m I_{n} - J_m
\end{array} \right) + E_{n,n} + E_{n+1,n+1} - E_{n+1,n} - E_{n,n+1}
\]
where $E_{i,j}$ is the matrix with a single 1 in the $(i,j)$-th place. A general statement is possible for a slight generalization to $B_{n,m}$, an $n$-clique and $m$-clique joined by a bridge, but we 
restrict to $B_{n,n}$ for brevity.

\begin{corollary}
  Let $G=B_{n,n}$. The spanning tree given by the two copies of the star graph $K_{1,n-1}$ together with the bridge (every edge has weight $1$) is a 2-sparsifier.
\end{corollary}

The result follows immediately from Theorem \ref{thm:examples} together with the spectral information of $B_{n,n}$ which is recorded in Table~\ref{tab: barbell spectrum}. We have not seen this information recorded in the literature elsewhere.

\begin{table}[h]
   \begin{tabular}{c c c}
        eigenvalue  & dimension & spanning eigenvectors \\ \hline
        0 & 1 &  $\ones_{2n}$ \\
     $1 + \alpha$ & 1 & $(-\ones_{n-1}, \alpha, -\alpha, \ones_{n-1})$  \\
        $n$ & $2n-3$ & $\{e_i-e_{i+1} \}_{i=1}^{n-2} \,\cup \,\{e_{n+1 +i} - e_{n+2+i} \}_{i=1}^{n-2}$ \\ 
        & & $ \cup \, \{(\ones_{n-1}, 1-n, 1-n,\ones_{n-1})\} $\\
       $1 + \beta$ & 1 & $(-\ones_{n-1}, \beta, -\beta, \ones_{n-1})$\\
     \end{tabular}
    \caption{The spectral information of the Barbell Graph $B_{n,n}$ by increasing eigenvalue.  Here $-1 < \alpha =n/2 - \sqrt{ n^2 + 4(n-1)}/2 < 0$ and $\beta = n/2 + \sqrt{ n^2 + 4(n-1)}/2 > n.$ }
     \label{tab: barbell spectrum}
\end{table}

The second family of examples will be given by the lollipop graph.
We define the $(m,p)$-Lollipop Graph to be an $m$-clique joined to the path graph on $p$ vertices by a bridge.  The $(7,5)$-Lollipop Graph is exhibited in Figure~\ref{fig: lollipop ex2} (and was also seen in Figure~\ref{fig: lollipop ex}). 
\begin{center}
    \begin{figure}[h!]
\includegraphics[scale = .3]{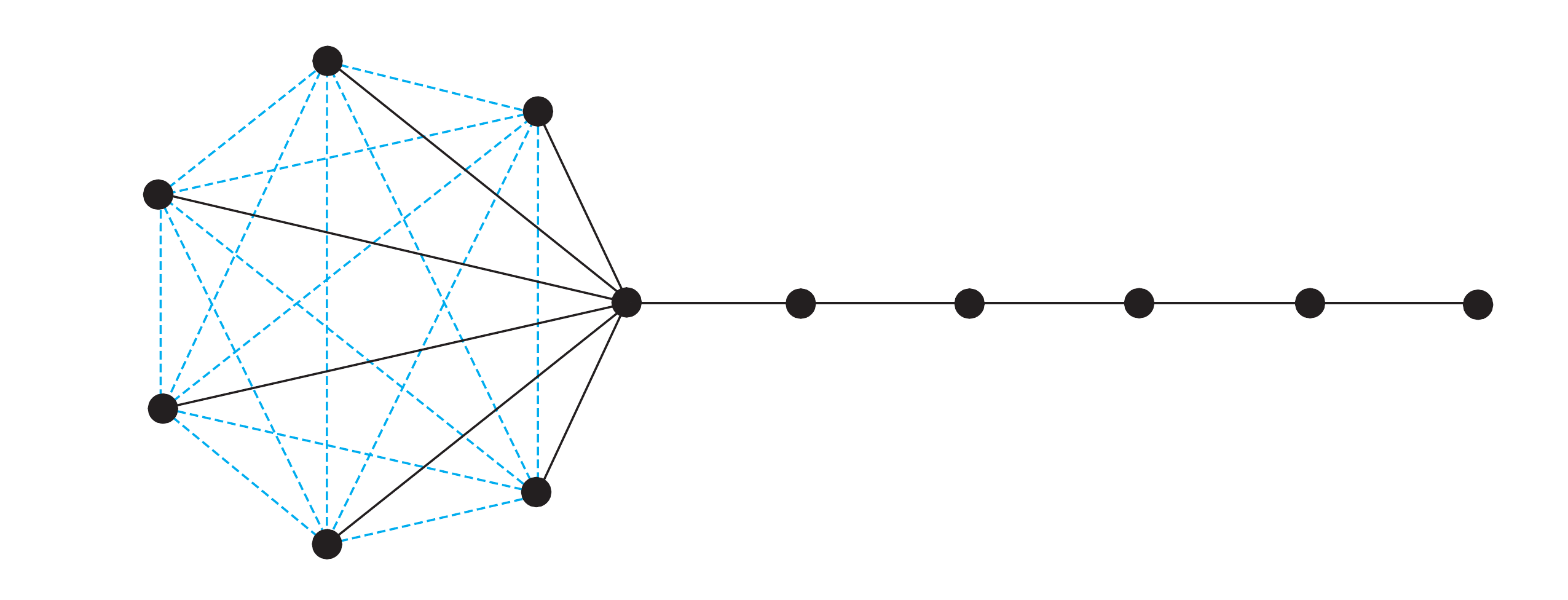}
    \caption{Lollipop graph on $n=12$ vertices and $m=26$ edges has a spanning tree $3$-sparsifier.} \label{fig: lollipop ex2}
        \end{figure}
\end{center}        
The spectrum of these graphs appears to be somewhat difficult to write down, however, the dynamics of the first few eigenvectors is easy to understand from a qualitative perspective. We recall that eigenvectors corresponding to small eigenvalues can really be understood as minimizers of the Rayleigh-Ritz functional
$$ \frac{\left\langle f, L_G f\right\rangle}{\left\langle f, f \right\rangle} = \frac{\sum_{(u,v) \in E} (f(u) -f(v))^2}{\sum_{v \in V} f(v)^2}.$$
In the case of the lollipop graph, the functional is easy to understand: variation across the path is `energetically' cheaper than variation within the complete graph since there are many more connections within the complete graph and variations show up in many more additional terms. Writing the vertex set as $V = K_{m-1} \cup P_{p+1}$, we expect the first few eigenvectors to be constant on $K_{m-1}$. The effect becomes more pronounced as the attached path becomes longer. Whenever that happens, Theorem \ref{thm:examples} immediately applies. Many other such examples can be constructed.

\section{Preserving the Quadratic Form}  \label{sec:preserv}

In the work of Batson, Spielman, Srivastava and Teng (see \cite{spiel3} and references therein) the sparsifiers of a graph $G$ are graphs $\tilde{G}$ on the same vertices whose quadratic forms $x^\top L_{\tilde{G}} x$  are approximately the same as $x^\top L_G x$. If one were to focus on the quadratic form, but think along the lines of this paper, then it is natural to consider a band-limited sparsifier of a graph $G$ as a subgraph $\tilde{G}$ such that 
$Q_G(x) = Q_{\tilde{G}}(x)$ for all $x \in \textup{span}\{\phi_1, \ldots, \phi_k\}$. 

\begin{definition} \label{def:Q(x)k sparsifier}
    A graph $\tilde{G}=([n],\tilde{E}, \tilde{w} )$ is a $Q_k(x)-${\bf sparsifier} of 
$G$ if $\tilde{E} \subseteq E$, 
$\tilde{w}_e > 0$ for all $e \in \tilde{E}$ and $Q_G(x) = Q_{\tilde{G}}(x)$ for all $x \in \textup{span}\{\phi_1, \ldots, \phi_k\}$. 
\end{definition}

Unlike in previous sections, we do not insist that a 
$Q_k(x)-$sparsifier of $G$ should lose edges. 
Definition~\ref{def:Q(x)k sparsifier} yields a polyhedron containing 
all possible Laplacians $L$ that correspond to $Q_k(x)-$sparsifiers of $G$. We derive this model below and compare it to Definition~\ref{def:k-sparsifier}. 
Note that the only $Q_n(x)-$sparsifier of $G$ is $G$ itself, but when $k < n$, this definition admits non-trivial sparsifiers.

 Recall that a spectrahedron is the intersection of a psd cone with an affine subspace of symmetric matrices.

\begin{lemma} \label{lem:SALambda is a spectrahedron}
For a given matrix $A \succeq 0$ and a proper subspace $\Lambda \subset \RR^n$, the set 
$$ S_A(\Lambda) = \{ B \succeq 0 \,:\, x^\top A x  =  x^\top B x \,\,\forall \, x \in \Lambda \}$$ 
is a spectrahedron. 
\end{lemma}
\emph{\nameref{pf:SALambda is a spectrahedron}}

The $Q_k(x)-$sparsifier set of $G$ is the intersection of 
$S_{L_G}(\textup{span}\{\phi_1, \ldots, \phi_k\})$ and the set of Laplacians of weighted subgraphs of $G$. When $A = L_G$, the proof of 
Lemma~\ref{lem:SALambda is a spectrahedron} shows that 
\[ S_{L_G}(\textup{span}\{\phi_1, \ldots, \phi_k\}) = \left\{L \succeq 0: \begin{array}{c}
     \langle L, \phi_i \phi_i^\top \rangle = \lambda_i \text{ for  } i=1, \ldots, k \\
       \langle L, \phi_i \phi_j^\top + \phi_j \phi_i^\top \rangle = 0\text{ for  } i\neq j
\end{array} \right\}. \]

The additional conditions needed for $L \in S_{L_G}(\textup{span}\{\phi_1, \ldots, \phi_k\})$ to be the Laplacian of a  weighted subgraph of $G$ are:
\begin{enumerate}
\item $L_{ij} = 0$ for all $i \neq j, \,(i,j) \not \in E$. This guarantees that $\tilde{E} \subseteq E$.
\item $L_{ij} \leq 0$ for all $i \neq j$. This guarantees that $\tilde{w}_{ij} = -L_{ij} \geq 0$.
\item $L_{ii} = - \sum_{(i,j) \in E} L_{ij}$ for all $i$.
\end{enumerate}
 Also, $\tilde{G}$ is connected if and only if $\rk(L) = n-1$.
We can bake conditions (1)--(3) into the following structured symbolic matrix:

\begin{align}
L({\tt w}) = 
\begin{bmatrix} 
\sum_{j \neq 1} {\tt{w}}_{1j} & - {\tt{w}}_{12} & -{\tt{w}}_{13} & \cdots & -{\tt{w}}_{1n}\\
-{\tt{w}}_{12} & \sum_{j \neq 2}  {\tt{w}}_{2j}& - {\tt{w}}_{23} & \cdots & -{\tt{w}}_{2n}\\
\vdots & \vdots & \vdots & \vdots & \vdots \\
-{\tt{w}}_{1n} & \cdots & \cdots & \cdots & \sum_{j \neq n} {\tt{w}}_{nj}
\end{bmatrix}
\end{align}
with ${\tt{w}}_{ij} = 0$ if  $(i,j) \not \in E$. Then all subgraphs of $G$ 
will have Laplacians of the form $L(\bar{w})$ for some $\bar{w} \in \RR^E_{\geq 0}$. 
By construction $L({\tt{w}}) \ones = 0 = L({\tt{w}}) \phi_1$. Using this, and the fact that 
$L({\tt{w}})$ is symmetric, we get the following facts:
\begin{lemma} \label{lem:facts about Lw}
\begin{enumerate}
\item $\phi_1^\top L({\tt{w}}) \phi_1 = 0$, 
\item $\phi_1^\top L({\tt{w}}) \phi_j = 0 = \phi_j^\top L({\tt{w}}) \phi_1$ for all $j > 1$, and 
\item $\langle L({\tt{w}}), \phi_i \phi_j^\top \rangle = \langle L({\tt{w}}), \phi_j \phi_i^\top \rangle$ which implies that 
\begin{align}
\langle L({\tt{w}}), \phi_i \phi_j^\top + \phi_j \phi_i^\top \rangle = 2 \langle L({\tt{w}}), \phi_i \phi_j^\top \rangle = 2 \langle L({\tt{w}}), \phi_j \phi_i^\top \rangle.
\end{align}
\end{enumerate}
\end{lemma}

 Define $L'({\tt{w}})$ to be the principal submatrix of $L({\tt{w}})$ obtained by deleting the last row and column. 

\begin{theorem} \label{thm:last}
The Laplacians of $Q_k(x)-$sparsifiers of $G$ are the matrices $L(w)$ for all $w$ in the polyhedron
$$P^{Q(x)}_G(k) = 
    \left\{ w \in \RR^E_{\geq 0} \,:\, 
    \begin{array}{l}
\langle L(w), \phi_i \phi_i^\top \rangle = \lambda_i \,\,\forall \, i =2, \ldots, k \\
\langle L(w), \phi_i \phi_j^\top  \rangle = 0, \,\,\forall \, 2 \leq i < j \leq k \\
\end{array} 
\right\}$$
which lies in the semialgebraic region $\det(L'(w)) \geq 0$. 
\begin{enumerate}
    \item The sparsity patterns of the $Q_k(x)-$sparsifiers of $G$
are in bijection with the faces of $P^{Q(x)}_G(k)$. 
    \item The connected sparsifiers correspond to faces where $\det(L'(w)) > 0$.
\end{enumerate}
\end{theorem}
\emph{\nameref{pf:last}}

Lemma~\ref{lem:low spec equality implies quadratic form equality} implies the following.


\begin{lemma}
    The set $\textup{Sp}_G(k)$ of $k-$isospectral subgraphs of $G$ is contained in the polyhedron 
    $P^{Q(x)}_G(k)$.
\end{lemma}


A disadvantage of the the larger $P^{Q(x)}_G(k)$ is that the spectrum of $\tilde{G}$ can be wildly different from that of $G$, even when restricting our attention to connected sparsifiers. We illustrate this on the 3-dimensional cube graph.   We avoid the notation $Q_d$ for cube graphs here to avoid confusion with the $Q_k(x)-$sparsifier notation. Recall from Theorem~\ref{thm:cube} that the 3-dimensional cube graph has no $4-$sparsifiers. The following example contrasts this with the set of $Q_4(x)-$sparsifiers of the cube graph.

\begin{example} \label{ex: cube Qk(x)}
Let $G = ([8], E) $ be the 3-dimensional cube graph. There is a spanning tree $Q_4(x)-$sparsifier of $G$ which preserves no eigenpair of $G$. In particular, the second eigenvalue of this tree is approximately $ .3677$, whereas for $G$, $\l_2 =2$.  
\end{example}

\begin{figure}[h]
    \begin{align*}
&    \begin{tikzpicture}
        \draw[thick] (0,0) -- (2,0) -- (2,2) -- (0,2) -- (0,0) -- (1.2,.7) -- (3.2,.7) -- (3.2,2.7) -- (1.2,2.7) -- (1.2,.7);
\draw[thick] (2,0)  -- (3.2,.7);
\draw[thick] (2,2)  -- (3.2,2.7);
\draw[thick] (0,2)  -- (1.2,2.7);
\draw  (0,0)  node[circle, fill = white,inner sep=2pt,draw] {1};
\draw  (2,0)  node[circle, fill = white,inner sep=2pt,draw] {2};
\draw  (0,2)  node[circle, fill = white,inner sep=2pt,draw] {3};
\draw  (2,2)  node[circle, fill = white,inner sep=2pt,draw] {5};
\draw  (1.2,.7)  node[circle, fill = white,inner sep=2pt,draw] {4};
\draw  (3.2,.7)  node[circle, fill = white,inner sep=2pt,draw] {6};
\draw  (1.2,2.7)  node[circle, fill = white,inner sep=2pt,draw] {7};
\draw  (3.2,2.7)  node[circle, fill = white,inner sep=2pt,draw] {8};
\end{tikzpicture}
 &&  \begin{tikzpicture}
        \draw[line width = 1.1mm ]   (1.2,.7)  -- (0,0) -- (2,0) ;
        \draw[line width = .6mm] (0,2) -- (0,0);
        \draw[thick]  (2,0) -- (2,2) ;
        \draw[thick]  (0,0) -- (1.2,.7) -- (3.2,.7) -- (3.2,2.7) ;
        \draw[thick] (0,2)  -- (1.2,2.7);
\draw  (0,0)  node[circle, fill = white,inner sep=2pt,draw] {1};
\draw  (2,0)  node[circle, fill = white,inner sep=2pt,draw] {2};
\draw  (0,2)  node[circle, fill = white,inner sep=2pt,draw] {3};
\draw  (2,2)  node[circle, fill = white,inner sep=2pt,draw] {5};
\draw  (1.2,.7)  node[circle, fill = white,inner sep=2pt,draw] {4};
\draw  (3.2,.7)  node[circle, fill = white,inner sep=2pt,draw] {6};
\draw  (1.2,2.7)  node[circle, fill = white,inner sep=2pt,draw] {7};
\draw  (3.2,2.7)  node[circle, fill = white,inner sep=2pt,draw] {8};
\end{tikzpicture}
    \end{align*}
    \caption{The 3-cube graph and a spanning tree $Q_4(x)-$sparsifier of it. Thin edges have weight 1, the medium edge has weight 2, and the thickest edges have weight 3.}
    \label{fig: cube Qk sparsifier}
\end{figure}
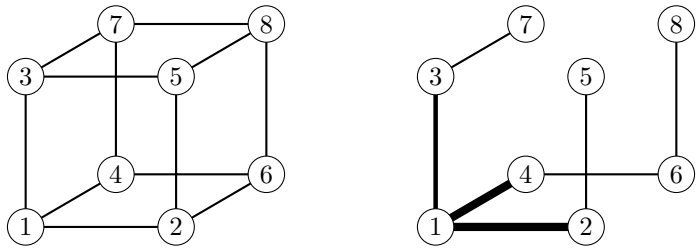

To see this, label the vertices of $G$ as in Figure~\ref{fig: cube Qk sparsifier}. The eigenvalues of $L_G$ are 
$$0^{(1)}, 2^{(3)} , 4^{(3)}, 6^{(1)},$$ where the exponents record multiplicity.
The following vectors form a orthonormal basis for the eigenspace of $\l_2 =2$ of $L_G$. 
\begin{align*}
    &\phi_2 =   [ 1  \,\, ,\,\,  -1  \,\, ,\,\,   1   \,\, ,\,\,  1 \,\, ,\,\,  -1  \,\, ,\,\,  -1   \,\, ,\,\,  1   \,\, ,\,\, -1]^\top / \sqrt{8}\\
   &  \phi_3 =[ 1   \,\, ,\,\,  1   \,\, ,\,\, -1  \,\, ,\,\,   1   \,\, ,\,\, -1\,\, ,\,\,     1  \,\, ,\,\,  -1  \,\, ,\,\,  -1]^\top  / \sqrt{8}\\
   &  \phi_4 =  [1   \,\, ,\,\,  1   \,\, ,\,\,  1   \,\, ,\,\, -1  \,\, ,\,\,   1  \,\, ,\,\,  -1   \,\, ,\,\, -1  \,\, ,\,\,  -1]^\top  / \sqrt{8}.
\end{align*}
From Theorem~\ref{thm:last}, we have the polyhedron
\[ P_G^{Q(x)}(4) = \left\{ w\in \RR^E_{\geq 0} : \begin{array}{l}
    w_{12}  + w_{35}+  w_{46} + w_{78} = 4 \\
    w_{13}  + w_{25}+  w_{47} + w_{68} = 4 \\
      w_{14}  + w_{26}+  w_{37} + w_{58} = 4 
\end{array} \right\} .\]
A valid choice of edge weights is \begin{align*}
   & w_{12} = w_{14} = 3 \\
    & w_{13} = 2 \\
    & w_{25} =  w_{37} = w_{46} = w_{68} = 1 \\
    &   w_{26} =  w_{35} = w_{47} = w_{58} = w_{78} = 0 .
\end{align*}
This spanning tree sparsifier $\tilde G$ is shown in Figure~\ref{fig: cube Qk sparsifier}.  Its Laplacian eigenvalues are
\[        0,
    0.3677,
    0.6383,
    1.3889,
    2.4974,
    3.6368,
    4.3896,
   11.0814,  \]
 and none of the eigenvectors of $\tilde G$ are eigenvectors of $G$. We can confirm, however, that the quadratic form is indeed preserved on the appropriate subspace:  \begin{align*}
    & \phi_2^\top L_{\tilde G} \phi_2 =  \phi_3^\top L_{\tilde G} \phi_3=  \phi_4^\top L_{\tilde G}\phi_4 =2  \text{, and}\\
    & \phi_i^\top L_{\tilde G} \phi_j =  0 , \,\,\forall \, i\neq j \in \{2,3,4\}. 
 \end{align*}

\begin{example} \label{ex:triangle}
Consider the weighted complete graph $G = K_3$ with edge weights 
$w_{12}=1, w_{13}=2, w_{23}=2$. Then 
\begin{align}
L_{G} = \begin{bmatrix} 3 & -1 & -2 \\ -1 & 3 & -2 \\ -2 & -2 & 4\end{bmatrix}
\end{align}
has eigenpairs:

$$ \left( \lambda_1 = 0, \phi_1 = \frac{1}{\sqrt{3}}\begin{bmatrix}1\\1\\1\end{bmatrix} \right), 
\,\, \left( \lambda_2 = 4, \phi_2 = \frac{1}{\sqrt{2}}\begin{bmatrix}1\\-1\\0\end{bmatrix} \right), 
\,\, \left( \lambda_3 = 6, \phi_3 = \frac{1}{\sqrt{6}}\begin{bmatrix}1\\1\\-2\end{bmatrix} \right).$$
Suppose we consider all graphs $\tilde{G}$ with Laplacian
\begin{align}\label{eq:3x3 Laplacian} 
L = \begin{bmatrix} a+b &  -a &  -b \\ -a & a+c & -c \\ -b & -c & b+c\end{bmatrix} 
\end{align}
and $Q_{\tilde{G}}(x) = Q_G(x)$ for all $x \in \textup{span} \{ \phi_1, \phi_2 \}$. Then 

\begin{align}
P^{Q(x)}_G(2)=   \left\{ (a,b,c) \in \RR^3_{\geq 0} \,:\,4a+b+c =8\right\}
\cong 
\left\{ (a,b) \in \RR^2 \,:\, 
\begin{array}{l}
a\geq 0 \\
b \geq 0 \\
8-4a-b \geq 0
\end{array} 
\right\}
\end{align}
which is the triangle with corners  
$\{(0,0),(0,8),(2,0)\}$ shown in Figure~\ref{fig:polytopesK3}. 
Each point in the triangle corresponds to a $Q_2(x)-$sparsifier of $G$.
\begin{figure}
\includegraphics[scale=0.3]{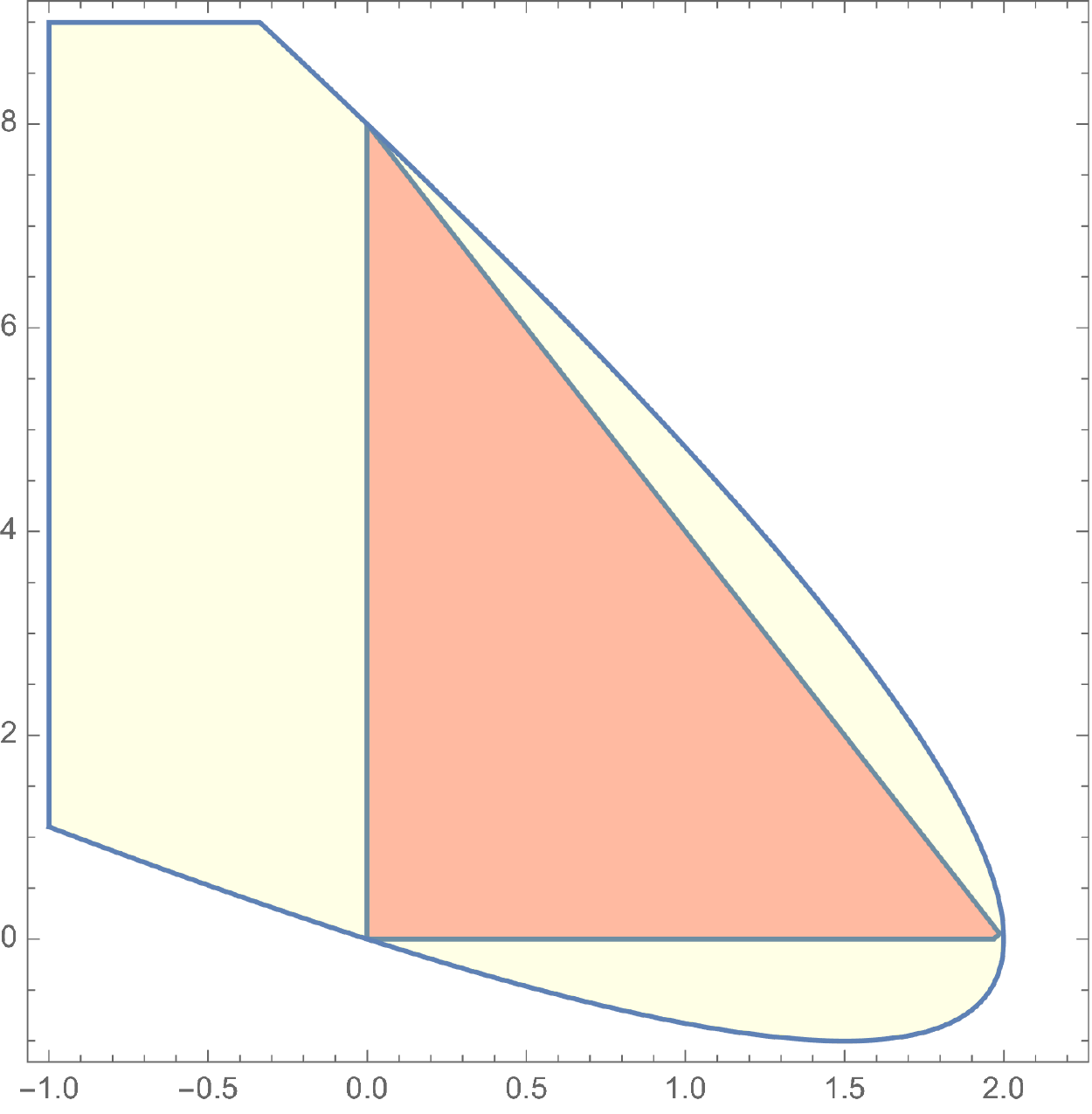} 
\quad \quad 
\includegraphics[scale=0.3]{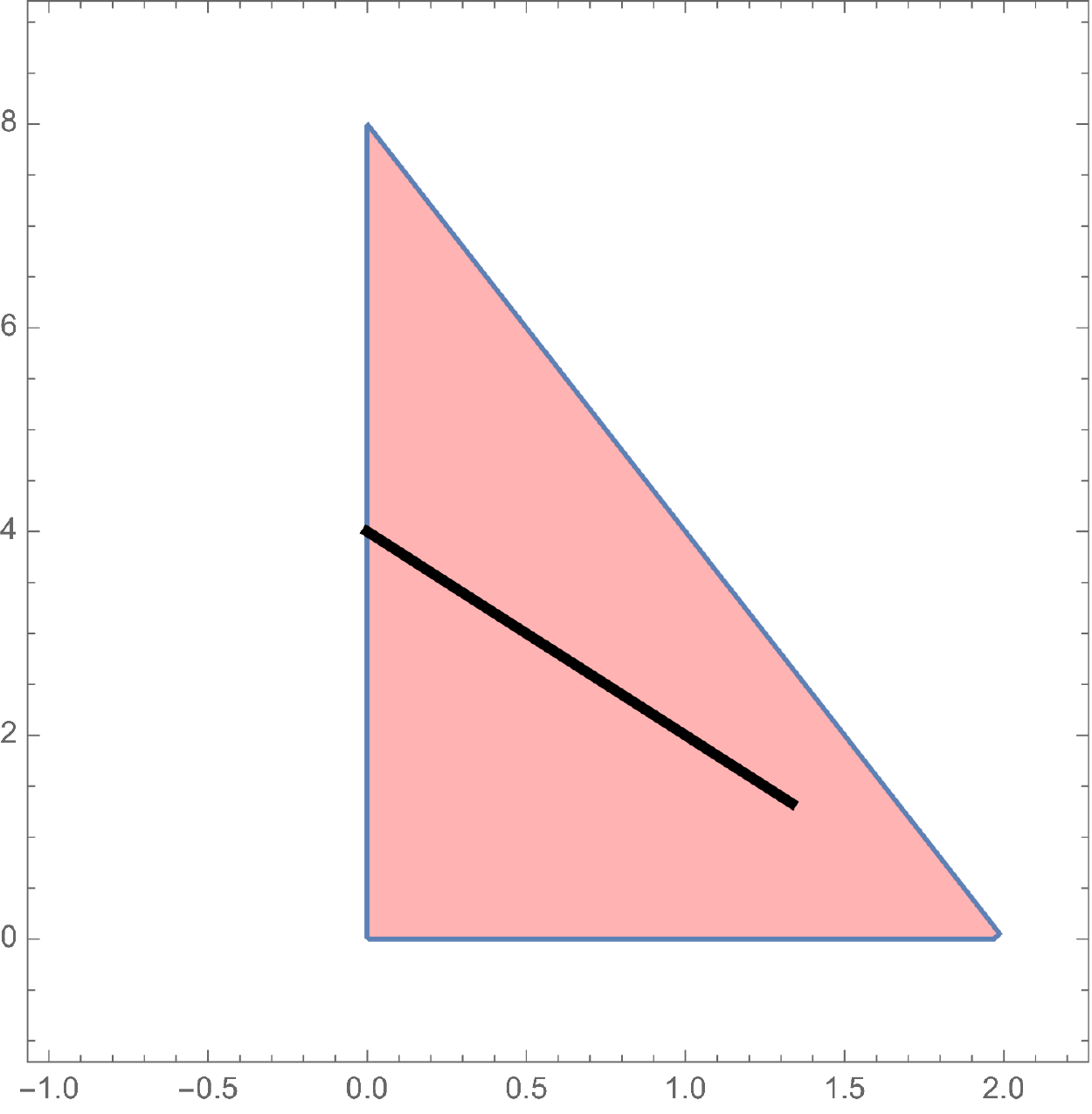}
\caption{Left: Triangle $P^{Q(x)}_{G}(2)$ of all graphs $\tilde{G}$ with  
$Q_G(x) = Q_{\tilde{G}}(x)$ for all $x \in \textup{span}\{\phi_1, \phi_2\}$. 
The triangle is contained inside the parabolic region $\det(L'(w)) \geq 0$. Right: Triangle $P^{Q(x)}_{G}(2)$ with the line segment $\Sp_G(2)$ inside it. \label{fig:polytopesK3}}
\end{figure}
The triangle is contained in the parabolic region cut out by 
$\det(L'(w)) = -4a^2 - b^2 - 4ab +8a + 8b \geq 0$. Its boundary intersects the triangle at its three corners where $\rank(L) =1$ and $\tilde{G}$ is disconnected. 
The sides of the triangle except for the corners correspond to $Q_2(x)-$sparsifiers of $G$ that are spanning trees.

On the other hand, the $2-$isospectral subgraphs of $G$ 
have Laplacians: 
$$ L = \begin{bmatrix} \frac{16+y}{6} & \frac{y-8}{6} & \frac{-2y-8}{6} \\ 
\frac{y-8}{6} & \frac{16+y}{6} & \frac{-2y-8}{6} \\ 
\frac{-2y-8}{6}  & \frac{-2y-8}{6} & \frac{16+4y}{6}\end{bmatrix}, \quad 0 \leq y \leq 8.$$
Equating to \eqref{eq:3x3 Laplacian}, $\Sp_G(2) = [(0,4,4),(\frac{4}{3},\frac{4}{3},\frac{4}{3})] \subset P^{Q(x)}_{G}(2)$, and $L_G$ corresponds to  
$$ (a,b,c) = (1,2,2) = \frac{1}{4} (0,4,4)+ \frac{3}{4}(4/3,4/3,4/3).$$ 
The point $(1,1,3) \in P^{Q(x)}_{G}(2) \setminus \Sp_G(2)$ corresponds to the graph $G'$ with 
$$ L_{G'} = \begin{bmatrix} 
2 & -1 & -1 \\ -1 & 4 & -3 \\ -1 & -3 & 4
\end{bmatrix}$$
and eigenpairs
$$(0,\psi_1 = \frac{1}{\sqrt{3}}(1,1,1)^\top), (3,\psi_2 = \frac{1}{\sqrt{6}}(2,-1,-1)^\top), (7,\psi_3 = \frac{1}{\sqrt{2}}(0,-1,1)^\top).$$
Therefore, $G'$ is not $2-$isospectral to $G$, but $Q_{G}(x) = Q_{G'}(x)$ for all $x \in \textup{span}\{ \phi_1, \phi_2\}$. Indeed, a linear combination of $\phi_1, \phi_2$ looks like $x = (\alpha + \beta, \alpha - \beta, \alpha)^\top$. Check that you get $8 \beta^2$ when you plug in $x$ into both  
$$Q_{G}(x) = 4 (\phi_2^\top x)^2 + 6 (\phi_3^\top x)^2 \textup{ and } Q_{G'}(x) = 3 (\psi_2^\top x)^2 + 7 (\psi_3^\top x)^2.$$
\end{example}

\section{Proofs} \label{sec:proofs}

\subsection*{Proof of Lemma~\ref{lem:low spec equality implies quadratic form equality}.}
\label{pf:low spec equality implies quadratic form equality}
If $x \in \textup{span}\{ \phi_1, \ldots, \phi_k\}$, then $x = \sum_{j=1}^k \beta_j \phi_j$ for some 
$\beta_j \in \RR$. Therefore, 
\begin{align*}
Q_G(x) = \sum_{i=2}^n \lambda_i  (\phi_i^\top x)^2 = \sum_{i=2}^n \lambda_i  (\phi_i^\top \sum_{j=1}^k \beta_j \phi_j)^2 
= \sum_{i=2}^n \lambda_i  (\sum_{j=1}^k \beta_j \phi_i^\top \phi_j)^2 = \sum_{i=1}^k \lambda_i \beta_i^2.
\end{align*}
On the other hand, 
\begin{align*}
    Q_{\tilde{G}}(x) = \sum_{i=2}^k \lambda_i (\phi_i^\top x)^2 + \sum_{j=1}^{n-k} \mu_j (\psi_j^\top x)^2.
\end{align*}
Since every $\psi_j$ is orthogonal to every $\phi_i$ we have that the second sum 
\begin{align*}
    \sum_{j=1}^{n-k} \mu_j (\psi_j^\top x)^2 = \sum_{j=1}^{n-k} \mu_j (\psi_i^\top \sum_{j=1}^k \beta_j \phi_j)^2 = 0. 
\end{align*}
Therefore, 
\begin{align*}
    Q_{\tilde{G}}(x) = \sum_{i=1}^k \lambda_i \beta_i^2 = Q_G(x).
\end{align*}
\qed

\subsection*{Proof of Theorem \ref{thm:k sparsifier}}
\label{pf: k-sparsifier}

Denote the eigenpairs of a $k-$isospectral 
subgraph $\tilde{G}$ as 
\begin{align} \label{eq:eigenpairs of sparsifier}
    (0,\phi_1), (\lambda_2, \phi_2), \ldots, (\lambda_k, \phi_k), (\mu_1, \psi_1), \ldots, (\mu_{n-k}, \psi_{n-k})
\end{align}
where 
the first $k$ are the same as those of $G$. 
The Laplacian of $\tilde{G}$ then has the form 
$$ L = \Phi_k \Lambda_k \Phi_k^\top + \Psi D \Psi^\top = \sum_{i=2}^k \lambda_i \phi_i \phi_i^\top + 
\sum_{j=1}^{n-k} \mu_j \psi_j \psi_j^\top$$
where $\Psi = \begin{bmatrix} \psi_1 & \cdots & \psi_{n-k} \end{bmatrix}$  and $D=\Diag(\mu_1, \ldots, \mu_{n-k})$ are variables that satisfy the following conditions:
\begin{enumerate}
    \item $\Phi_k^\top \Psi = 0$ since each $\phi_i$ must be orthogonal to each $\psi_j$. 
    \item $\Psi^\top \Psi = I_{n-k}$ since $\{\psi_j\}$ must be a set of orthonormal vectors.
    \item $\mu_j \geq \lambda_k$ for all $j$ since we want $\tilde{G}$ to have the same first $k$ eigenvalues as $G$.
\end{enumerate}
A ready-made candidate for $\Psi$ that satisfies conditions (1)--(2) is 
\begin{align}
\Phi_{> k} :=  \begin{bmatrix} \phi_{k+1} & \cdots & \phi_n \end{bmatrix}.
\end{align}
Any other $\Psi$ is of the form 
$\Phi_{>k} B$ for an orthogonal matrix $B$ in $\RR^{(n-k) \times (n-k)}$ since 
$\{\psi_j\}$ is an orthogonal basis for 
$\textup{Span}\{\phi_{k+1}, \ldots, \phi_n\}$.
Letting $X = B D B^\top$ we can rewrite $L$ as 
\begin{align*}
     L &= \Phi_k \Lambda_k \Phi_k^\top + (\Phi_{>k} B) D ( \Phi_{>k}B)^\top \\
     &= \Phi_k \Lambda_k \Phi_k^\top + \Phi_{>k} (B D B^\top) \Phi_{>k}^\top = \Phi_k \Lambda_k \Phi_k^\top + \Phi_{>k} X \Phi_{>k}^\top .
\end{align*}
The set of all matrices of the form $BDB^\top$ where $B$ is an orthogonal matrix and $D$ is diagonal with all entries positive is the set of positive definite matrices of size $n-k$, namely  the interior of the psd cone 
$\mathcal{S}^{n-k}_+$. Since we want $D_{ii} \geq  \lambda_k$, we 
require $D \succeq \lambda_k I_{n-k}$ which implies that $X = BDB^\top \succeq \lambda_k I_{n-k}$, or equivalently, 
$X = \lambda_k I_{n-k} + Y, \,\,\, Y \succeq 0.$ 
Hence 
\begin{align}
\Phi_{>k} X \Phi_{>k}^\top = \Phi_{>k} (\lambda_k I_{n-k} + Y) \Phi_{>k}^\top = 
\lambda_k \Phi_{>k} \Phi_{>k} ^\top + 
\Phi_{>k} Y \Phi_{>k}^\top.
\end{align}
Putting all of the above together, the Laplacians of $k-$isospectral subgraphs $\tilde{G}$ of $G$ must be a subset of matrices of the form 
\begin{align} \label{eq:L}
    L = \underbrace{\Phi_k \Lambda_k \Phi_k^\top + \lambda_k \Phi_{>k} \Phi_{>k}^\top}_{=:F} 
    + \Phi_{>k} Y \Phi_{>k}^\top.
\end{align}
where $Y \succeq 0$. 
By construction, any $L$ of the form \eqref{eq:L} has $n-1$ positive eigenvalues and hence 
its rank is $n-1$. The matrix $F$ has  
    spectral decomposition 
    $$F = \Phi \Diag(0,\l_2, \ldots, \l_k, \l_k, \ldots, \l_k) \Phi^\top, $$
and hence, $F \in \mathcal{S}^n_+$ and 
$\rank(F)=n-1$. Since $L$ is obtained by adding $F$ to a psd matrix $\Phi_{>k} Y \Phi_{>k}^\top \in \mathcal{S}^n_+$, we have that $L \in \mathcal{S}^n_+$. 

For $L$ in \eqref{eq:L} to be the Laplacian of a subgraph of $G$ we impose the  conditions:
\begin{align} 
    L_{st} \leq 0 \,\,\forall \, (s,t) \in E \label{eq:lin conds1}\\
    L_{st} = 0 \,\,\forall \, s \neq t, \, (s,t) \not \in E \label{eq:lin conds2}
\end{align}
The only property left to check is that $L$ is now 
weakly diagonally dominant.

By construction, $(0,\ones)$ is the first eigenpair of $L$ which means that $L\ones = 0$. 
This guarantees that $-L_{ii}$ is the sum of the off-diagonal entries in row $i$. Since, $L \succeq 0$, $L_{ii} \geq 0$, and by \eqref{eq:lin conds1} and \eqref{eq:lin conds2}, $L_{ii}$ is the only potential positive entry in row $i$. If $L_{ii} = 0$ then all entries in row $i$ are $0$ since $L \succeq 0$ which means that vertex $i$ in the subgraph $\tilde{G}$ associated to $L$ is isolated. However, this is impossible since $k \geq 2$ which means that we are requiring that $\l_2 > 0$ is the second eigenvalue of 
$L$. Thus, $\Sp_G(k)$ is contained in the set of Laplacians of $k-$isospectral subgraphs of $G$. 

On the other hand, $L_{\tilde{G}}$ for any $k-$isospectral subgraph $\tilde{G}$ of $G$ is in $\Sp_G(k)$; choose $Y = \Diag(\mu_{k+1}-\l_k, \ldots, \mu_n-\l_k)$. This 
proves Theorem~\ref{thm:k sparsifier}.
\qed

\subsection*{Proof of Corollary \ref{cor:23}}
\label{pf: 23}
    Every face of $P_G(k)$ is determined by some collection of inequalities from $\{ L_{st} \leq 0 \,:\, (s,t) \in E\}$ holding at equality. 
    The graph $\tilde{G}$ is missing the edge $(s,t) \in E$ if and only if $L_{st}=0$. Thus the faces of $P_G(k)$ that are contained in $\textup{Sp}_G(k)$ index all possible sparsity patterns of $k-$sparsifiers of $G$.
\qed

\subsection*{Proof of Theorem \ref{thm: Sp2 independent of basis}}
  \label{pf: Sp2 independent of basis}
Recall that for a fixed eigenbasis $\Phi$, 
 \[\Sp_G(k) = \{ L = \Phi_{k} \L_k  \Phi_{k}^\top  + \l_k\Phi_{>k} \Phi_{>k}^\top + \Phi_{>k} Y \Phi_{>k}^\top   \},\]
 where $Y \succeq 0$, and the edges of the graph provide additional constraints on valid choices of $Y$.  
 The `fixed' matrix in this expression is 
 \[F = \Phi_{k} \L_k  \Phi_{k}^\top  + \l_k\Phi_{>k} \Phi_{>k}^\top.\] 
We first show that $F$ does not depend on the choice of basis. 
 
  Let $L_G$ have $m$ eigenspaces, and let $ 0 < \l_{i_2} < \ldots < \l_{i_m}$ be its distinct positive eigenvalues where $i_j$ is the index of the first eigenvalue associated to the $j$-th eigenspace. If $\Sp_G(k)$ preserves the first $\ell\in \{2,\ldots, m-1\}$ positive eigenspaces, then  $k = i_{\ell+1} - 1$.
Let $\Phi$ and $\Psi$ be two eigenbases of  $L_G$,  and let $\Phi_{\l}$ denote the submatrix of $\Phi$ consisting of the columns which span the eigenspace with  eigenvalue $\l$.  Then there is an orthogonal matrix $U_j$ such that $\Phi_{\l_{i_j}}  U_j = \Psi_{\l_{i_j}}$ for each $j \in [m]$.
Since $U_jU_j^\top$ is the identity, we have that
\begin{align*}
    F & =  \Phi_{k} \L_k  \Phi_{k}^\top  + \l_k\Phi_{>k} \Phi_{>k}^\top \\
    & =   \sum_{j=2}^{\ell} \l_{i_j} \Phi_{\l_{i_j}}   \Phi_{\l_{i_j}}^\top  + \l_{i_\ell}  \sum_{j=\ell+1}^{m}   \Phi_{\l_{i_j}}   \Phi_{\l_{i_j}}^\top \\
    & =  \sum_{j=2}^{\ell} \l_{i_j} \Phi_{\l_{i_j}}  U_j U_j^\top \Phi_{\l_{i_j}}^\top  + \l_{i_\ell}  \sum_{j=\ell+1}^{m}   \Phi_{\l_{i_j}}    U_j U_j^\top \Phi_{\l_{i_j}}^\top \\
    & = \sum_{j=2}^{\ell} \l_{i_j} \Psi_{\l_{i_j}} \Psi_{\l_{i_j}}^\top  + \l_{i_\ell}  \sum_{j=\ell+1}^{m}   \Psi_{\l_{i_j}}   \Psi_{\l_{i_j}}^\top  \\
    & = \Psi_{k} \L_k  \Psi_{k}^\top  + \l_k\Psi_{>k} \Psi_{>k}^\top
\end{align*}
Thus $F$ does not depend on the choice of eigenbasis. Now, let $\tilde {U} = \Diag( U_{\ell+1}, \ldots, U_m)$, and note that 
\begin{align*}
  \Psi_{>k} Y \Psi_{>k}^\top & =   \Phi_{>k} \tilde {U} Y\tilde {U}^\top \Phi_{>k}^\top,
\end{align*}
and moreover that $Y \succeq 0$ if and only if $ \tilde {U} Y \tilde {U}^\top \succeq 0$.
Thus $ L =  F + \Psi_{>k} Y \Psi_{>k}^\top $if and only if $L = F + \Phi_{>k} \tilde {U} Y\tilde {U}^\top \Phi_{>k}^\top.$ 
Hence the set $\Sp_G(k)$ is independent of the choice of eigenbasis.
\qed

\subsection*{Proof of Lemma \ref{lem: F for one fixed eigenval}}
\label{pf: F for one fixed eigenval}
Recall that the Laplacian of $K_n$ is $L_{K_n} = nI_n - J_n$. 
If $\l_2 = \ldots = \l_k $, then 
\begin{align*}
    F = \Phi_k \Lambda_k \Phi_k^\top + \lambda_k \Phi_{>k} \Phi_{>k}^\top  = \lambda_2 \sum_{i=2}^n \phi_i \phi_i^\top = \l_2 (I_n - \phi_1 \phi_1^\top) = \l_2 \left(I_n- \frac{1}{n} J_n\right).
    \end{align*}
    \qed

\subsection*{Proof of Theorem~\ref{lem: weighted complete sparsifies}}
\label{pf: weighted complete sparsifies}
    Since a weighted, complete graph $G$ has no missing edges,
    $$\Sp_G(k)= \{ L = \Phi_k \Lambda_k \Phi_k^\top + \lambda_k \Phi_{>k} \Phi_{>k}^\top + \Phi_{>k} Y \Phi_{>k}^\top \,:\, L_{st} \leq 0 \,\,\forall \,\, s \neq t , \,\,Y \succeq 0\}.$$
    No $k-$isospectral subgraph will sparsify if and only if none of the hyperplanes $L_{st}=0$ intersect the psd cone 
    $\mathcal{S}^{n-k}_+$, or equivalently, $\mathcal{S}^{n-k}_+$ lies in the open halfspace $L_{st} < 0$ for all $s \neq t$.
    Recall that $L_{st} \leq 0$ is $\langle Y, c_s c_t^\top \rangle \leq - F_{st}$ where $c_i$ is the $i$th column of 
    $\Phi_{> k}^\top$. Suppose there is some 
    $\bar{Y} \succeq 0$ such that $\langle \bar{Y}, c_s c_t^\top \rangle > 0$. Then for a fixed $Y_0 \succeq 0$ and 
    $\alpha > 0$ 
    \begin{align}
        \langle Y_0 + \alpha \bar{Y} , c_s c_t^\top \rangle = c_s^\top Y_0 c_t + \alpha (c_s^\top \bar{Y} c_t) \longrightarrow 
        \infty \textup{ as } \alpha \longrightarrow \infty.
    \end{align}
    Therefore, for large enough $\alpha$, the psd matrix $Y_0 +\alpha \bar{Y}$ violates $L_{st} \leq 0$. This means 
    that $L_{st} \leq 0$ is not valid on all of $\mathcal{S}^{n-k}_+$ and the plane 
    $L_{st}=0$ cuts through $\mathcal{S}^{n-k}_+$.

    By the above argument, if  $\mathcal{S}^{n-k}_+$ lies in the open halfspace $L_{st} < 0$, then it must be that 
    $\langle Y, c_s c_t^\top \rangle \leq 0$ for all $Y \succeq 0$. The polar of the psd cone $\mathcal{S}^{n-k}_+$ is the set of 
    all matrices $X$ such that $\langle Y,X \rangle \geq 0$ for all $Y \in \mathcal{S}^{n-k}_+$. It is well known that the psd cone 
    is the polar dual of itself, and hence $\langle Y, c_s c_t^\top \rangle \leq 0$ for all $Y \succeq 0$ if and only if 
    $-c_sc_t^\top$ lies in the polar dual cone, or equivalently, $-c_sc_t^\top$ is a psd 
     matrix. Since $-c_sc_t^\top$ is a rank one matrix, it is psd if and only if $c_t = -\beta c_s$ for some $\beta \geq 0$. 
    We conclude that if $\mathcal{S}^{n-k}_+$ lies in the open 
    halfspace $L_{st} < 0$ for all $(s,t) \in E$, then there are scalars 
    \begin{align} \label{eq:needed condn}
    \beta_{st} \geq 0 \textup{ such that } c_t = -\beta_{st} c_s \,\,\forall \, s \neq t.
    \end{align}
    
If $k \leq n-2$, then $\Phi_{> k}^\top$ has at least two rows and at least one nonzero column (say $c_1$). Since all edges 
$(1,t)$ are present in $G$, if $\eqref{eq:needed condn}$ holds, then $c_2, \ldots, c_n$ are nonpositive multiples of $c_1$. This makes 
the rows of $\Phi_{>k}$ dependent which contradicts that the rows are part of an eigenbasis. 

    Therefore, when $k \leq n-2$, some hyperplane $L_{st} =0$ cuts through the psd 
    cone and in particular, supports $\textup{Sp}_G(k)$. Points lying on this face of $\textup{Sp}_G(k)$ are $k-$sparsifiers which miss the edge $(s,t)$.
\qed\\

In the above proof, if $k = n-1$ then $\Phi_{> n-1} = \phi_n$ which has at least two nonzero entries (say the first and second) since $\phi_n^\top \ones = 0$. If it has a third nonzero entry as well (say the third), then the second and third entries are negative multiples of the first but then the third is not a negative multiple of the second which contradicts \eqref{eq:needed condn} and we get that some $L_{st}=0$ supports $\Sp_{G}(k)$ and $G$ has a $(n-1)-$sparisifer. 

The requirement that $\phi_n$ has at least three nonzero entries is a type of genericity condition. Recall from the $K_5$ example that if $\phi_n$ has only two nonzero entries then a $(n-1)-$sparsifier can fail to exist. 
In the generic case, we can also give a direct constructive proof that shows the existence of a $(n-1)-$sparsifier.

\begin{proposition} \label{prop:n-1 sparsifier} 
Let $G=(V,E,w)$ be a complete graph where every edge has a positive weight. If the largest eigenvalue $\lambda_n$ has an eigenvector $\phi_n$ with at least three nonzero entries, then there is a $(n-1)-$sparsifier deleting at least one edge.
\end{proposition}
\begin{proof} We start by considering the Laplacian $L_G = D-A$ of $G$. Since none of the edge weights vanish, $L_G$ has positive entries on the diagonal and negative entries everywhere else. Our goal is to construct a sparsifier $H$ that preserves the first $k=n-1$ eigenvectors and eigenvalues. Since the last eigenvector is isolated, $\lambda_n \geq \lambda_{n-1}$, we have very concrete knowledge of how the Laplacian $L_H$ has to look like: the difference $L_G-L_H$, interpreted as a linear operator, can only act on the eigenspace corresponding to $\phi_n$. Thus $L_H$ is a rank one perturbation of $L_G$ and, for some constant $c \in \mathbb{R}$
$$L_H = L_G + c (\phi_n \phi_n^{\top}).$$
Moreover, we observe that the largest eigenvalue then obeys
$$ \lambda_n(L_H) = \lambda_n(L_G) + c.$$
In particular, this shows that by setting $c \geq 0$ we are guaranteed to obtain a new Laplacian whose first $n-1$ eigenvalues and eigenvectors coincide exactly with those of $L_G$. It remains to check whether $L_H$ corresponds to a Laplacian of a weighted graph: for this, we require that the matrix is symmetric, each row sums to zero, and that the off-diagonal entries are all non-negative. Since the graph is connected, $\phi_1$ is constant. By orthogonality, this implies that the entries of $\phi_n$ sum to 0. Moreover, if $\phi_n$ has at least three non-zero coordinates, then there are at least two with the same sign which gives rise to at least one off-diagonal entry of the matrix $(\phi_n \phi_n^{\top})$. This shows that there exists a choice of $c>0$ such that $L_H$ has at least one off-diagonal entry that is zero while all other off-diagonal entries are negative and leads to an $(n-1)-$sparsifier.
\end{proof}

We note that for `generic' weights and `generic' eigenvectors, the above procedure will typically result in an $(n-1)-$sparsifier that deletes \textit{exactly} one edge. Recalling the Linear Algebra Heuristic with $k=n-1$, we see that
    $$|E| \leq {n \choose 2} - {n-k+1 \choose 2} = {n \choose 2} - {n-(n-1)+1 \choose 2} = {n \choose 2} - 1,$$
and this coincides exactly with the previous proof: one can hope to erase a single edge but not more than that.

\subsection*{Proof of Theorem \ref{thm:cube}} 
\label{pf:cube} We argue by contradiction and assume that there exists a sparsification $Q_d'$ of $Q_d$ at $k=d+1$. This corresponds to an assignment
of weights to the $d \cdot 2^{d-1}$ edges $e \in E$ of $Q_d$ such that at least one of the edge weights is zero (which corresponds to the vanishing of an edge).
The spectrum and eigenvectors of the Laplacian of $Q_d$  are well studied.  
The eigenvectors are of the form $\phi_s (v) = (-1)^{s^\top v}$ for $s,v\in\{0,1\}^d$. The eigenvectors $\phi_s$ and $\phi_t$ have the same eigenvalue if and only if $\ones_d^\top s = \ones_d^\top t$. In particular, the eigenspace for $\lambda_2 = 2$ is spanned by the eigenvectors $\{\phi_{e_i}\}_{i=1}^d$.
We start by noting that, since the bottom of the spectrum of $Q_d$ is preserved in $Q_d'$ (both the eigenvalues and the eigenvectors), we have that for all functions $f:V \rightarrow \mathbb{R}$ on $Q_d'$ 
with mean value 0,
$$ \frac{ \sum_{(i,j) \in E} w_{ij}  (f(i) - f(j))^2}{ \sum_{i \in V} f(i)^2} \geq \lambda_2.$$
We first show that every single edge weight in $Q_d'$ has to be at least
$w_{ab} \geq 1$. Suppose $(a,b)$ is an edge of $Q_d'$ 
 such that $w_{ab} < 1$. Since $a, b \in \{0,1\}^d$ are adjacent in $Q_d$, there is exactly one coordinate $\ell \in [d]$ for which $a_\ell \neq b_\ell$. The eigenvector $\phi = \phi_{e_\ell}$ has eigenvalue 2 and is such that $ \phi(a) = (-1)^{a_\ell} = (-1) (-1)^{b_\ell}  = -\phi(b)$. We may assume without loss of generality that $\phi(a) =1$ and so, $\phi(b)=-1$. 
We will now introduce a new function $\psi:V \rightarrow \mathbb{R}$ as follows
$$ \psi(v) = \begin{cases} 1 + \varepsilon \qquad &\mbox{if}~v = a \\
-1 - \varepsilon \qquad &\mbox{if}~v = b\\
\phi(v) \qquad &\mbox{otherwise.}
\end{cases}$$
Then, $$ \sum_{i \in V} \psi(v) = 0.$$
We consider the Rayleigh-Ritz quotient and observe that
\begin{align*}
\frac{ \sum_{(i,j) \in E} w_{ij}  (\psi(i) - \psi(j))^2}{ \sum_{i \in V} \psi(i)^2} &= \frac{ \sum_{(i,j) \in E} w_{ij}  (\psi(i) - \psi(j))^2}{4 \varepsilon + 2 \varepsilon^2 + 2^d}.
\end{align*}
To analyze the numerator, we split the sum into four different sums: (1) the edge $(a,b)$, (2) the other edges incident to $a$, (3) the other edges incident to $b$ and (4) edges incident to neither $a$ nor $b$. We obtain
\begin{align*}
\sum_{(i,j) \in E} w_{ij}  (\psi(i) - \psi(vj)^2 &= 
w_{ab}(\psi(a) - \psi(b))^2 
+ \sum_{(i,j) \in E \atop { i,j \neq a,b }} w_{ij}  (\psi(i) - \psi(j))^2\\
&+\sum_{(a,j) \in E \atop { j\neq b}} w_{aj}  (\psi(a) - \psi(j))^2 
+ \sum_{(b,j) \in E \atop {  j\neq a}} w_{bj}  (\psi(b) - \psi(j))^2.
\end{align*}
The first term is simple:
\begin{align*}
  w_{ab}(\psi(a) - \psi(b))^2  &= w_{ab} (2+ 2\varepsilon)^2 \\
  &= 4 w_{ab} + 8 w_{ab} \varepsilon+ 4 w_{ab} \varepsilon^2 \\
  &=   w_{ab}(\phi(a) - \phi(b))^2  + 8 w_{ab} \varepsilon+ 4 w_{ab} \varepsilon^2.
\end{align*}
The second sum is also easy to analyze: the edges are not incident to $a$ or $b$ and thus the function $\psi$ coincides with the function $\phi$ for all terms and
\[\sum_{(i,j) \in E \atop { i,j \neq a,b }} w_{ij}  (\psi(i) - \psi(j))^2 = \sum_{(i,j) \in E \atop { i,j \neq a,b }} w_{ij}  (\phi(i) - \phi(j))^2.\]
For an edge $(a,j)$ where $j \neq b$, since $a$ and $j$ do not differ in the 
$\ell$-th coordinate, we have that $\psi(j) = \phi(j) = \phi(a) = 1$. Therefore, 
\begin{align*}
 \sum_{(a,j) \in E \atop { j\neq b}}  w_{aj}  (\psi(a) - \psi(j))^2 &= \sum_{(a,j) \in E \atop { j\neq b}} w_{aj}  \varepsilon^2 = \varepsilon^2 \sum_{(a,j) \in E\atop { j\neq b}} w_{aj}
\end{align*}
and, via the same reasoning
\begin{align*}
 \sum_{(b,j) \in E \atop {  j\neq a}} w_{bj}  (\psi(b) - \psi(j))^2 = \varepsilon^2 \sum_{(b,j) \in E \atop {  j\neq a}} w_{bj}.
\end{align*}
Therefore,
\begin{align*}
\sum_{(i,j) \in E} w_{ij}  (\psi(i) - \psi(j))^2 &=  \sum_{(i,j) \in E} w_{ij}  (\phi(i) - \phi(j))^2 \\
&+ 8 \omega_{ab} \varepsilon + 4 \omega_{ab} \varepsilon^2 \\
&+ \varepsilon^2 \left( \sum_{(a,j) \in E} w_{aj} + \sum_{(b,j) \in E} w_{bj} \right), 
\end{align*}
Note that the first sum on the right of the equal sign is over all edges in $E$. Indeed, the second sum from the four sums above accounts for all edges 
not adjacent to $a$ or $b$. The term $w_{ab}(\phi(a) - \phi(b))^2$ from the first of the four sums contributes the edge $(a,b)$. Further, since 
$\phi(a)=\phi(j)$ when $(a,j) \in E$ and $j \neq b$, there is no harm in 
adding the sum of terms $w_{aj}(\phi(a)-\phi(j))^2$ for all $(a,j) \in E, \,j \neq b$, and similarly, also the sum of the terms 
$w_{bj}(\phi(b)-\phi(j))^2$ for all $(b,j) \in E, \,j \neq a$.
Therefore, 
$$ \frac{ \sum_{(i,j) \in E} w_{ij}  (\psi(i) - \psi(j))^2}{ \sum_{i \in V} \psi(i)^2} =   \frac{8 w_{ab} \varepsilon + \mathcal{O}(\varepsilon^2) +  \sum_{(i,j) \in E} w_{ij}  (\phi(i) - \phi(j))^2}{ 4\varepsilon + \mathcal{O}(\varepsilon^2) + \sum_{i \in V} \phi(i)^2}
$$
Since $\phi$ assumes values in $\left\{-1,1\right\}$ and corresponds to eigenvalue 2, we have
$$  \sum_{i \in V} \phi(i)^2 = 2^d \quad \mbox{and} \quad \sum_{(i,j) \in E} w_{ij}  (\phi(i) - \phi(j))^2 = 2^{d+1}.$$
Thus, 
$$ \frac{ \sum_{(i,j) \in E} w_{ij}  (\psi(i) - \psi(j))^2}{ \sum_{i \in V} \psi(i)^2} =   \frac{2^{d+1} + 8 w_{ab} \varepsilon + \mathcal{O}(\varepsilon^2) }{ 2^d + 4\varepsilon + \mathcal{O}(\varepsilon^2)}.
$$
Differentiating the expression on the right in $\varepsilon$ at $\varepsilon =0$, we obtain
$$ \frac{d}{d\varepsilon} \quad \frac{2^{d+1} + 8 w_{ab} \varepsilon + \mathcal{O}(\varepsilon^2) }{ 2^d + 4\varepsilon + \mathcal{O}(\varepsilon^2)} \bigg|_{\varepsilon =0} = - \frac{8 (1 - w_{ab})}{2^d} < 0.
$$ 
This shows that, for $\varepsilon > 0$ sufficiently small,
$$ \frac{ \sum_{(i,j) \in E} w_{ij}  (\psi(i) - \psi(j))^2}{ \sum_{i \in V} \psi(i)^2} < \lambda_2$$ 
which contradicts the Courant-Fischer theorem. Thus $w_{ab} \geq 1$ for all edges. 

Suppose now $w_{ab} > 1$ for some $(a,b) \in E$. Letting $\ell$ be the coordinate on which $a$ and $b$ differ, the eigenvector $\phi = \phi_{e_\ell}$ of $\l_2$ which has $\phi(a) = - \phi(b)$ shows that
$$ 2 = \lambda_2 = \frac{ \sum_{(i,j) \in E} w_{ij}  (\phi(i) - \phi(j))^2}{ \sum_{i \in V} \phi(i)^2} >  \frac{ \sum_{(i,j) \in E}  (\phi(i) - \phi(j))^2}{ \sum_{i \in V} \phi(i)^2}  = 2$$
which is a contradiction. Therefore all the weights are necessarily $w_{ab} = 1$.
\qed

\subsection*{Proof of Theorem \ref{thm: graphs with no n-2 sparsifiers}}
\label{pf: graphs with no n-2 sparsifiers}
    We will construct these graphs from a specifically chosen spectrum, which we show in Table~\ref{tab: no n-2 sparsifier spectrum}.  
    \begin{table}[h]
        \begin{tabular}{c|c}
        distinct eigenvalues &  eigenspaces \\ \hline
         $ \l_1 = 0 $  &  $ \Lambda_1 = \spanset \{\phi_1' = \ones\}$\\
            $ \l_2 = n$  &  $\spanset\{\phi_1', \phi_{n-1}', \phi_n'\}^\perp$ \\
             $ \l_{n-1 } = 3n $ &  $ \Lambda_{n-1} = \spanset\{\phi_{n-1}' = e_{n-1} - e_n \}$\\ 
              $ \l_n= 7n +6 $ &  $ \Lambda_n = \spanset\{\phi_n' = (0, \ldots, 0, 2, -1 ,-1)\}$
        \end{tabular}
        \caption{The spectrum of a weighted graph with $n$ vertices that has no $(n-2)-$sparsifier. }
        \label{tab: no n-2 sparsifier spectrum}
    \end{table}
    Let $\{\phi_i\}$ be an orthonormal basis chosen from these eigenspaces with $\phi_1, \phi_{n-1}, \phi_n$ being the normalized $\phi_1', \phi_{n-1}', \phi_n'$, and let $\L = \Diag( 0\, ,\, n \ones^\top_{n-3} \, ,\, 3n \, ,\, 7n + 6).$
 The resulting matrix $\Phi \L \Phi^\top =  L_G $ defines the graph $G = ([n], E, w)$ with edge weights
\[ w_{ij} = \begin{cases}
 1     &  i\in [n-3], j > i \\
    2n +3&  i = n-2, j = n-1, n \\
 0     &  i = n-1, j = n 
\end{cases}.\]
See Figure~\ref{fig: graph cant sparsify} for a depiction of the graph when $n = 4. $ Because $\l_2 = \ldots = \l_{n-2},$ by Lemma~\ref{lem: F for one fixed eigenval}, Laplacians in $\Sp_G(n-2)$ have the form
    \[n I_n - J_n +\Phi_{>n-2} Y \Phi_{>n-2}^\top, \,\, Y \succeq 0. \]
Letting $Y = \begin{bmatrix}
    a  & b \\
     b & c
\end{bmatrix}$, we see that
 \[ (\Phi_{>n-2} Y \Phi_{>n-2}^\top)_{ij} = 
 \begin{cases}
     0   &i \in [n-3], j \in [n] \\
   b/\sqrt{3} - c/3   &i = n-2, j = n-1  \\
     -b/\sqrt{3} - c/3   &i = n-2, j = n \\
    c/6 - a/2    &i = n-1, j = n
 \end{cases}.
\] 
Most entries of this matrix are zero due to the structure of $\phi_{n-1}$ and $\phi_n$; only three edges of the graph have flexibility.
Because $(n-1,n) \not \in E(G)$, we must have that 
\[ (L_G)_{n-1,n} = (n I_n - J_n +\Phi_{>n-2} Y \Phi_{>n-2}^\top)_{n-1,n} = 0,\] which is to say $-1 + c/6 - a/2 = 0.$
Thus we can write \[ Y = \begin{bmatrix}
    c/3 -2  & b \\
     b & c
\end{bmatrix}.\]
  The psd constraints on $Y$ are then 
\[ c \geq 6, \quad  |b| \leq \sqrt{c^2/3 -2c}.\]
The remaining two edge inequalities defining the sparsifier set from edges $(n-2, n-1)$ and $(n-2,n)$ respectively are
\[  -1 - b/\sqrt{3} - c/3 \leq 0  , \,  \qquad  -1 +  b/\sqrt{3}- c/3  \leq 0.\] 
We can consolidate these as $|b| \leq \sqrt{3} + c /\sqrt{3}$.
Thus the psd constraints on $Y$ are strictly contained in the polyhedron defined by the edge inequalities:
\begin{align*}
    |b| \leq \sqrt{c^2/3-2c}\text{ and } c \geq 0 &\implies |b| \leq c/\sqrt{3}  \implies |c| <  c/\sqrt{3}  + \sqrt{3} .
\end{align*} 
Therefore no edge can be deleted by a psd choice of $Y$. 
We note that while these graphs have no $(n-2)-$sparsifiers, the set of $(n-2)-$isospectral graphs is not a point -- even more than that, it is full dimensional and unbounded. Any matrix $Y$ satisfying the psd constraints provides a valid $(n-2)-$isospectral graph.

\begin{figure}[h]
    \centering
\includegraphics[scale = .4]{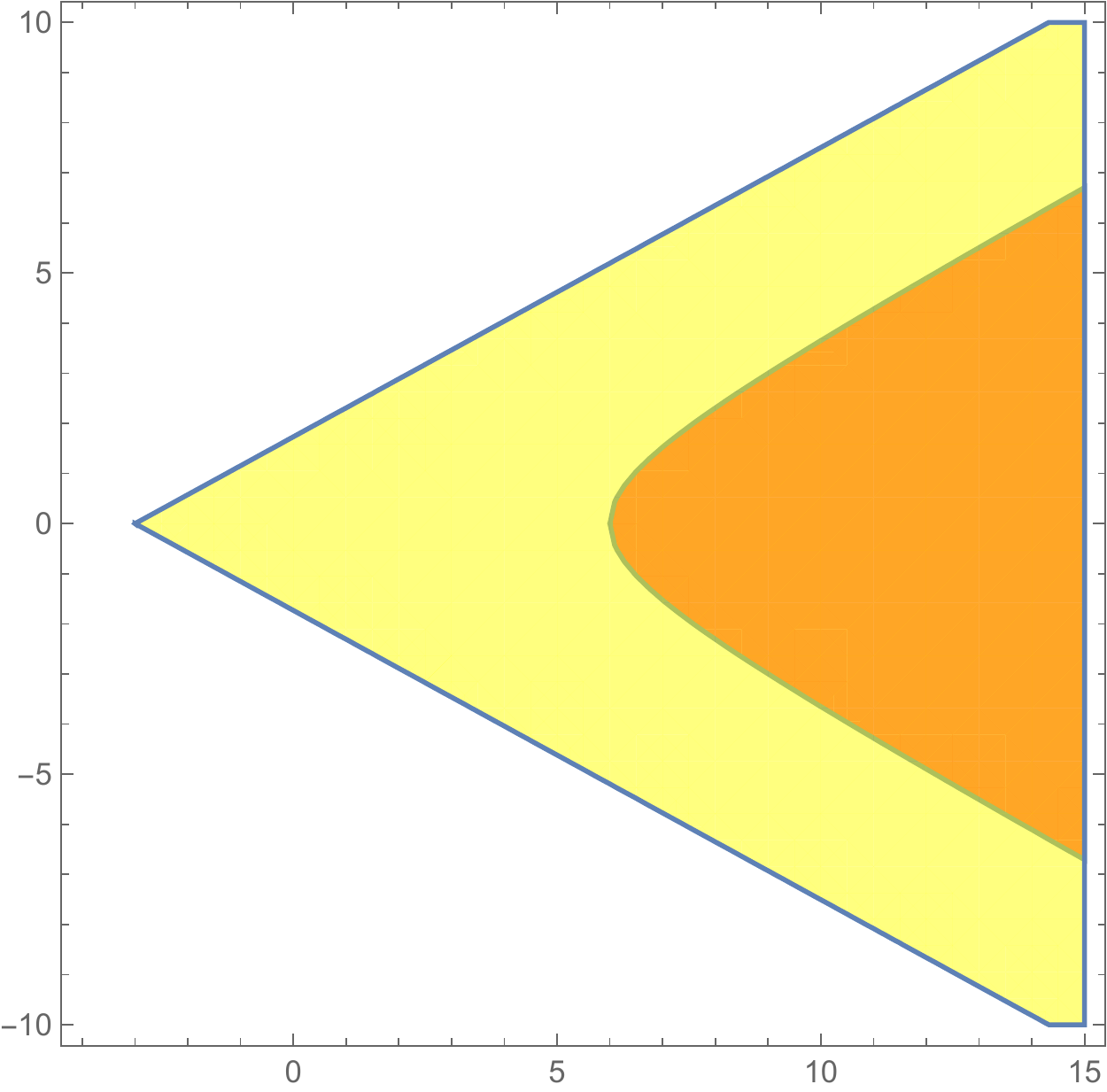}
    \caption{The valid $(b, c)$ for $(n-2)-$isospectral graphs of the graph family in Theorem~\ref{thm: graphs with no n-2 sparsifiers} are in orange. The yellow region is the polyhedron defined by the edge inequalities, and the orange region is defined by the psd constraints on $Y$. Because these regions are strictly nested, these graphs cannot lose any edges.}
    \label{fig: cant sparsify, Y region}
\end{figure}
\qed

We note that the choice $\l_2 = \ldots = \l_{n-2}$ in the above proof serves only to make the argument cleaner.  There are likely many choices of eigenvalues and eigenvectors for which a similar construction produces a graph that does not sparsify. 

\subsection*{Proof of Theorem~\ref{thm: unweighted complete minus one edge, SP2}}

\label{pf: unweighted complete minus one edge, SP2}
        We can assume the missing edge is $(n-1,n)$. This graph has three eigenspaces. 
        The first two eigenpairs are $(0, \ones)$, $(n-2, e_{n-1} - e_n), $ and the eigenspace for $\l = n$ is spanned by the eigenvectors $$\{e_1 - e_j\}_{j=2}^{n-2} \cup\{ \ones_n - n(e_{n-1} + e_n)/2\}.$$ 
        Let $\Phi$ be any orthonormal eigenbasis of $L_G$ so that $\phi_n = (\ones_n - n e_{n-2})/ \sqrt{ n(n-1)}$. (This is a convenient choice for $\phi_n$, different from the ones listed. Check that it is an eigenvector with eigenvalue $n$.)
        Then, the Laplacians in $\Sp_G(2)$ look like 
        $$L = (n-2)(I_n - J_n/n) + \Phi_{>2} Y \Phi_{>2}^\top$$ 
        where $Y \succeq 0$. Restricting our attention to matrices $Y$ of the form $\Diag(0, \ldots, 0, y)$, we get Laplacians of the form
\[ L = (n-2)(I_n - J_n/n) + y \phi_n \phi_n^\top  :\,\, y \geq  0 \]
        
The condition $L_{n-1,n} = 0 $  implies that $(y \phi_n \phi_n^\top)_{n-1,n} = y /(n(n-1)) = (n-2)/n , $       that is, $ y = (n-2)(n-1).$ Plugging this in, we get the Laplacian of the star graph $K_{1,n}$ with equal edge weights $n-2$, where the center of the star is the vertex $n-2$.
       This choice of central vertex is not special -- to place vertex $j\in [n-2]$ at the center of the star, set $\phi_n = (\ones_n - ne_{j})/ \sqrt{ n(n-1)}$.
\qed

\subsection*{Proof of Theorem \ref{thm: Kn spanning tree}}
\label{pf: Kn spanning tree}
   By Corollary~\ref{cor:nestedness}, it suffices to show that the result holds for $k= n-1$. Consider a fixed orthonormal basis of $\RR^n$ where $\phi_1 = \ones/ \sqrt{n}$ and the last eigenvector is $\phi_n = ( -\ones_{n-1}, n-1)/ \sqrt{n (n-1)}$. 
By Lemma~\ref{lem: F for one fixed eigenval},
 \begin{align*}
    \textup{Sp}_{K_n}(n-1) = 
    \left\{ L= nI - J_n + y  \phi_{n} \phi_{n}^\top : y \geq 0, L_{st} \leq 0  \,\, \forall \,\, s \neq t \right\}.
    \end{align*} 
Additionally,
\[ \phi_{n} \phi_{n}^\top = \frac{1}{n(n-1)} \begin{bmatrix}
    J_{n-1}  & (1-n) \ones_{n-1} \\
    (1-n) \ones_{n-1}^\top & (n-1)^2  
\end{bmatrix}.\]
For $s < t$ we see that 
\[ L_{st} = \begin{cases}  -1 + y/(n(n-1)) & t \leq n-1 \\
-1 -y/n  & t = n 
\end{cases}. \]
The conditions $L_{st} \leq 0$ and $ y \geq 0$ imply that $0 \leq y \leq n(n-1)$.
Therefore we can simplify again:
      \begin{align*}
    \textup{Sp}_{K_n}(n-1) = 
    \left\{ nI - J + y  \phi_{n} \phi_{n}^\top : 0 \leq  y\leq n(n-1) \right\}.
    \end{align*} 
    Choosing $y = n(n-1)$  produces the Laplacian \[ L = \begin{bmatrix}
        n I_{n-1} & - n \ones_{n-1} \\
        - n \ones_{n-1}^\top & n(n-1)
    \end{bmatrix}.\]
    This corresponds to the star graph where the vertex $n$ is at the center, and every edge has weight $n$.   There is nothing special about the choice of vertex $n$ as the center of the star, by relabeling the vertices any vertex can be made the center. 
\qed

\subsection*{Proof of Theorem \ref{thm:wheel}}
\label{pf: wheel}

We start by recording the eigenvalues and an eigenbasis of the wheel graph in Table~\ref{tab:Wheel spectra}, which are well understood from the formulation of $W_{n+1}$ as a cone over the cycle $C_n$ \cite{BrouwerHaemersSpectra, merris}. We use the following notation for the nontrivial eigenvectors of $C_n$ arising from the representations of $\ZZ/n$, $j \in [\lfloor (n-1)/2 \rfloor]$. 
 \[ \phi_j :=\begin{bmatrix}
 1\\
 \cos(2\pi j/n)\\
 \vdots \\
 \cos(2\pi (n-1)j/n)
 \end{bmatrix}, 
 \psi_j :=
 \begin{bmatrix}
 0\\
 \sin(2\pi j/n)\\
 \sin(2\pi 2j/n)\\
 \vdots \\
 \sin(2\pi (n-1)j/n)
 \end{bmatrix}.\]

 \begin{table}[h!]
     \centering
     \begin{tabular}{c c c}
        Eigenvalue  & Multiplicity & Eigenvectors \\ \hline
         0 & 1 & $\ones_{n+1}$\\
         $\l_j =3 - 2\cos(2\pi j/n)$ & $2^*$ & $[\phi_j, 0], [\psi_j, 0]$\\
       $n+1$  & 1 & $[-\ones_{n}, n]$
     \end{tabular}
     \caption{The spectrum of $W_{n+1}$, where $j = 2, \ldots,  \lfloor (n-1)/2 \rfloor$.   $^*$If $n$ is even, $\l_{n/2}$ has multiplicity 1 because $\phi_{n/2} = -\psi_{n/2}$. Here the index $j$ reflects an ordering natural to the cycle rather than by increasing eigenvalue. }
     \label{tab:Wheel spectra}
 \end{table}

We can also write the Laplacian of $W_{n+1}$ in terms of that of $C_n$.
     \[ L_{W_{n+1} } = 
\left(\begin{array}{@{}c|c@{}}
  \begin{matrix}
 3 & -1 & 0 &\ldots & 0 & -1 \\
  -1 & 3 & -1 & 0 & \ldots & 0   \\
 &\ddots   & \ddots& \ddots &  \ddots  \\
  -1 &  0 &\ldots & 0 & -1 & 3 
  \end{matrix}
  &   \begin{matrix}
  -\ones_{n}
  \end{matrix} \\
\hline 
  \begin{matrix}  -\ones_{n}^\top
  \end{matrix}&
  \begin{matrix}
  n
  \end{matrix}
\end{array}\right) = \left(\begin{array}{@{}c|c@{}}
 L_{C_{n}} + I_{n}
  &   
  -\ones_{n}
  \\
\hline 
  -\ones_{n}^\top & 
  n
\end{array}\right),
\]    

Let $\Phi$ be any orthonormal eigenbasis of $L_{W_{n+1}}$, let $\l = 3 - 2 \cos (2\pi/n)$, and let 

$ \Lambda' = \Diag( 0 ,\ones_{n-1}^\top,  n+1)$.  
We will show that \[ L= \l\Phi\Lambda' \Phi^\top \in \textup{Sp}_{W_{n+1}}(3),\] and moreover that $L = \l L_{K_{1,n}}$ is the Laplacian of the claimed spanning tree.
By construction, the first three eigenpairs of $L$ are $(0, \phi_1),(\l, \phi_2),  (\l, \phi_3)$ -- this agrees with the first three eigenpairs of $L_{W_{n+1}}$.  All other eigenvalues of $L$ are $\l$ or $  \l n$, which are both at least $\l$. Taking $\Lambda$ to be the matrix of eigenvalues of $L_{W_{n+1}}$,
\begin{align*}
   \Phi \Lambda' \Phi^\top & =     \Phi \L \Phi^\top - \Phi \left(\L - \Lambda' \right) \Phi^\top \\
& = L_{W_{n+1}} - \Phi \Diag(0, \l_2 -1, \ldots, \l_n -1, 0)  \Phi^\top \\
& = L_{W_{n+1}} - \left(\begin{array}{@{}c|c@{}}
 L_{C_{n}} 
  &   
  \bf{0}
  \\
\hline 
  \bf{0}^\top & 
  0
\end{array}\right)  =   \left(\begin{array}{@{}c|c@{}}
I_{n}
  &   
  -\ones_{n}
  \\
\hline 
  -\ones_{n}^\top & 
  n
\end{array}\right) = L_{K_{1,n}}.
\end{align*}
So we see that $L = \l L_{K_{1,n}}$, which satisfies all the constraints $L_{st} \leq 0$ for $s\neq t$, $L_{st} = 0$ for $st \notin E(W_{n+1})$.
Thus, $L$ defines a $3-$sparsifier of $L_{W_{n+1}}$.
\qed

\subsection*{Proof of Theorem \ref{thm:examples}}
\label{pf: example}
It is clear from the variational characterization and the quadratic form
$$ \left\langle f, L_G f \right\rangle = \sum_{(u,v) \in E} w_{uv} (f(u) - f(v))^2$$
that removing edges cannot increase any eigenvalue. 
We will now prove the statement for $k=2$, the general argument is then identical via Rayleigh-Ritz and the variational characterization.
We can therefore remove edges until we arrive at the spanning tree and compute its first non-trivial eigenvector $\phi_2$. If it is now the case that for all edges in $(u,v) \in E \setminus E_T$ that $\phi_2(u) = \phi_2(v)$, then it means that adding these edges back in does not change the value of the quadratic form. This shows that $\phi_2$ is also an eigenvector on $G$ corresponding to the same eigenvalue as on the tree since it minimizes the quadratic form. This proves that the tree is a $k=2$ sparsifier for the choice of basis of eigenspaces given by $\left\{ \ones, \phi_2 \right\}$. In general, we may now repeat the same argument for any $k$ and the result follows.
\qed

\subsection*{Proof of Lemma \ref{lem:SALambda is a spectrahedron}}
\label{pf:SALambda is a spectrahedron}
 Suppose $\Lambda = \textup{span}\{ \phi_1, \ldots, \phi_k\}$. Then 
$x \in \Lambda$ implies that $x = \sum_{j=1}^k \beta_j \phi_j$ for $\beta_j \in \RR$. Therefore,
\begin{align}
xx^\top = (\sum_{j=1}^k \beta_j \phi_j)( \sum_{j=1}^k \beta_j \phi_j^\top) = \sum \beta_i^2 \phi_i \phi_i^\top + \sum_{i \neq j} \beta_i \beta_j (\phi_i \phi_j^\top + \phi_j \phi_i^\top). 
\end{align}
Set $s_{ii} := \langle A, \phi_i \phi_i^\top \rangle \geq 0$ and 
$s_{ij} = \langle A, \phi_i \phi_j^\top + \phi_j \phi_i^\top \rangle \in \RR$. These are constants that we can compute from $A$ and a basis of $\Lambda$.
Then 
\begin{align}
x^\top A x &= \langle A, xx^\top \rangle = \sum \beta_i^2 \langle A, \phi_i \phi_i^\top \rangle + \sum_{i \neq j} \beta_i \beta_j \langle A, \phi_i \phi_j^\top + \phi_j \phi_i^\top \rangle \\
&= \sum \beta_i^2 s_{ii} + \sum_{i \neq j} \beta_i \beta_j s_{ij}.
\end{align}
Setting $y_{ii} := \langle B, \phi_i \phi_i^\top \rangle$ and 
$y_{ij} = \langle B, \phi_i \phi_j^\top + \phi_j \phi_i^\top \rangle \in \RR$ we have 
\begin{align*}
S_A(\Lambda)  = & \left \{ B \succeq 0 \,:\, Q_A(x) = Q_B(x) \,\,\forall \,\, x \in \Lambda \right \} \\
 = & \left \{ B \succeq 0 \,:\, \langle A, xx^\top \rangle = \langle B, xx^\top \rangle \,\,\forall \,\, x \in \Lambda \right \} \\
 = & \left \{ B \succeq 0 \,:\, \sum \beta_i^2 s_{ii} + \sum_{i \neq j} \beta_i \beta_j s_{ij} = 
\sum \beta_i^2 y_{ii} + \sum_{i \neq j} \beta_i \beta_j y_{ij} \,\,\forall \, \, \beta \in \RR^k \right \}\\
 = & \left \{ B \succeq 0 \,:\, y_{ii} = s_{ii}, y_{ij} = s_{ij} \right \}\\
 = & \left \{ B \succeq 0 \,:\, \langle B, \phi_i \phi_i^\top \rangle  = s_{ii}, \langle B, \phi_i \phi_j^\top + \phi_j \phi_i^\top \rangle  = s_{ij} \right \}.
\end{align*}
\qed

\subsection*{Proof of Theorem \ref{thm:last}}
\label{pf:last}
For any point $w \in \RR^E_{\geq 0}$, $L_G(w)$ is the Laplacian of a subgraph $\tilde{G}$ of $G$ by the definition 
of $L_G({\tt{w}})$. In particular, it is already psd and has the edge sparsity of $G$ built in. Therefore all points in 
$P^{Q(x)}_G(k)$ are $Q_k(x)-$sparsifiers of $G$. 
The faces of $P^{Q(x)}_G(k)$ correspond to a collection of inequalities $w_i \geq 0$ holding at equality. 
Therefore, the sparsest sparsifiers lie on the smallest dimensional faces of $P^{Q(x)}_G(k)$. 
It could be that some of the graphs on the boundary of $P^{Q(x)}_G(k)$ are disconnected but they still satisfy the needed conditions on the quadratic form. Connected subgraphs of $G$ must have Laplacians of 
rank $n-1$, and $\det(L'_G(w))=0$ if and only if 
$\rank(L_G(w)) < n-1$.
Since $L'_G(w)$ is a principal minor of 
$L_G(w)$, and $L_G(w) \succeq 0$ for all $w$ in the polyhedron $P^{Q(x)}_G(k)$, it must be that the polyhedron satisfies 
$\det(L_G'(w)) \geq 0$. 
\qed

\section{Conclusion} \label{sec:conc}
We conclude with a number of final remarks, comments and observations.\\

\textbf{A Dynamical Systems Motivation.} Instead of considering a graph, one could think about the behavior of dynamical systems \textit{on} graphs. A particularly natural example is the behavior of the heat equation: given a temperature $f: V \rightarrow \mathbb{R}$, 
one would naturally ask that vertices that are surrounded by warmer vertices should heat up while vertices surrounded by colder vertices should get colder. This suggests that the temperature $u:[0,\infty] \times V \rightarrow \mathbb{R}$ is initially given by the function $f$, meaning $u(0,v) = f(v)$ and, at time $t$, satisfies
$$
        \frac{\partial u}{\partial t}(v) = \sum_{(v,w) \in E} w_{vw} \cdot ( u(t,w) - u(t,v))
$$
which can be concisely written as $ \partial_t u = -(D-A) u$ or $u_t = -L u$. We also note since every edge is summed over twice, we have
$ \sum_{v \in V} u(t,v) = \sum_{v \in V} u(0,v)$ and the total amount of caloric energy in the graph remains constant.  Since $L$ is diagonalizable, we deduce that if 
$$ u(0,x) = \sum_{i=1}^{n} a_i \phi_i \qquad \mbox{then} \qquad u(t,x) = \sum_{i=1}^{n} a_i e^{-\lambda_i t} \phi_i$$
 which can be observed by noting that $e^{-\lambda_i t} \phi_i$ is a solution when the initial condition is given by $\phi_i$. The general case then follows from linearity.
 Since the exponential decay is larger for larger eigenvalues, we see that the behavior of $u(t,x)$ for large values of $t$ is, to leading order, well-approximated by 
$$ u(t,x) =  \sum_{i=1}^{n} a_i e^{-\lambda_i t} \phi_i \sim  \sum_{i \leq k}^{} a_i e^{-\lambda_i t} \phi_i$$
with an error at scale $\sim \exp(-\lambda_{k+1} \cdot t)$. In light of this, one can motivate the graph sparsification as one that preserves the long-time behavior of the heat equation as accurately as possible. \\

\textbf{A Cheeger Inequality Motivation.} Cheeger's inequality \cite{cheeger} shows that the eigenvalue $\lambda_2$ (the `algebraic connectivity') gives bounds on how easily the graph can be decomposed into two graphs with relatively few edges running between them. This can be seen as an extension of the basic algebraic fact that $\lambda_2 = 0$ if and only if $G$ is comprised of at least two disjoint graphs. Pushing the analogy, we know that $\lambda_m = 0$ if and only if $G$ is comprised of at least $m$ mutually disconnected graphs. One could now naturally wonder whether $\lambda_m$ can say anything about how easy or hard it is to decompose a graph into $m$ clusters with relatively few edges running between them. Results of this type have indeed recently been obtained, we refer to \cite{dey, kwok, lee, liu}. It is an immediate consequence of our sparsification approach that the number of connected components remains preserved once $k > m$.\\

\textbf{A Diffusion Map Interpretation.}
Laplacian eigenvectors of graphs have proven tremendously useful in obtaining low-dimensional embeddings $\phi: V \rightarrow \mathbb{R}^d$ that reflect the overall structure of the graph. Famous methods of this type are Laplacian eigenmaps \cite{niyogi} or diffusion maps \cite{coif}. By the nature of our sparsification,  the sparsifiers share the same low-dimensional embeddings. 
These mappings have been used successfully in dimensionality reduction exactly because these lower-dimensional embeddings tend to capture important information contained in the low-frequency part of the spectrum of the Laplacian. This could also be seen as an alternative (equivalent) motivation for our sparsification ansatz.\\

\textbf{Spectrally Extremal Examples.}
Theorem \ref{thm:cube} completely resolved the case of the cube $\left\{0,1\right\}^d$ and shows a very natural type of stability result: the first two eigenvalues (and the $(d+1)$-dimensional associated eigenspaces) fix the cube completely. This is particularly satisfying insofar as one would not expect there to be any particular canonical subgraph that shares many of the same properties and symmetries: the cube graph is already perfect just the way it is.  It would be interesting to see whether similar results exist for other families of graphs that arise in a similar fashion, in particular the example of Cayley graphs.\\

\textbf{Preserving other Laplacians.} Throughout the paper, our goal was to preserve the low frequency eigenvalues and eigenvectors of the Kirchhoff Laplacian $L= D-A$. However, there are several other notions of Graph Laplacian that have a number of different properties, examples being $I - D^{-1/2}AD^{-1/2}$ and $I - A D^{-1}$. 
They preserve different types of properties and emphasize a somewhat different aspect of graph geometry. The Laplacian $I - A D^{-1}$, for example, is intimately connected to the behavior of the random walk and sparsifying while preserving the low-frequency spectrum of $I - A D^{-1}$ would lead to another way of preserving the local geometry.
In these examples edge weights enter nonlinearly into the Laplacian while entering linearly into $L=D-A$. As such it is reasonable to assume that the case $L=D-A$ is somewhat distinguished and perhaps allows for the most complete analysis. Nonetheless, it would be interesting to see whether the main idea that underlies our ansatz could be carried out in other, more nonlinear, settings. \\

\textbf{Approximate Preservation.} The main philosophy that underlies our approach to sparsification is that small eigenvalues and eigenvectors encapsulate the overall global structure of a graph which suggests preserving them while making the graph more sparse. While there are some results of a purely algebraic nature, many of the results like Cheeger's inequality and statements of this nature, are continuous in the underlying parameters.  This suggests that it would be quite feasible to preserve low frequency eigenvalues and eigenvectors \textit{approximately}, although we don't develop it in this paper.

\vspace{10pt}
\noindent{\bf Acknowledgments.} S.S. is supported by the NSF (DMS-2123224) and the Alfred P. Sloan Foundation. R.T. is supported by the Walker Family Endowed Professorship at the University of Washington. C.B. was supported by the NSF (DMS-1929284) while in residence at ICERM in Providence, RI, during the Discrete Optimization Semester program. We thank Nikhil Srivastava for an inspiring conversation at the 2023 Joint Math Meetings.

\end{document}